\documentclass{article}
\PassOptionsToPackage{numbers,compress}{natbib}
\usepackage[main,final]{neurips_2026}

\usepackage[utf8]{inputenc} %

\usepackage{graphicx}

\usepackage[T1]{fontenc}

\usepackage[hidelinks,colorlinks]{hyperref}       %
\usepackage{url}            %
\usepackage{booktabs}       %
\usepackage{amsfonts}       %
\usepackage{nicefrac}       %
\usepackage{microtype}      %
\usepackage{pgf}

\IfFileExists{headers/config/showoverfull.config}{
	\overfullrule=1cm
}{
}

\usepackage{marginnote}

\usepackage[backgroundcolor=none,linecolor=red,textsize=footnotesize]{todonotes}

\usepackage{etoolbox}

\newbool{includeappendix}
\setbool{includeappendix}{true} %
\IfFileExists{headers/config/noappendix.config}{
	\setbool{includeappendix}{false}
}{}

\newif\ifincludeappendixx
\ifbool{includeappendix}{
	\includeappendixxtrue
}{
	\includeappendixxfalse
}

\usepackage{xr} %
\usepackage{filecontents}

\ifbool{includeappendix}{}{

	\externaldocument{appendix-labels}
}

\usepackage{xspace}

\newcommand{\eg}{e.g., }
\newcommand{\ie}{i.e., }

\newcommand{\llama}{\textsc{Llama3.1-8B}}
\newcommand{\qwen}{\textsc{Qwen3-30B}}
\newcommand{\mistral}{\textsc{Ministral-3-14B}}

\newcommand{\elifive}{\textsc{ELI5}}

\usepackage{acro} %

\usepackage{listings}

\usepackage{textcomp}

\usepackage{xcolor}

\usepackage[scaled=0.8]{beramono}

\definecolor{ckeyword}{HTML}{7F0055}
\definecolor{ccomment}{HTML}{3F7F5F}
\definecolor{cstring}{HTML}{2A0099}

\lstdefinestyle{numbers}{
	numbers=left,
	framexleftmargin=20pt,
	numberstyle=\tiny,
	firstnumber=auto,
	numbersep=1em,
	xleftmargin=2em
}

\lstdefinestyle{layout}{
	frame=none,
	captionpos=b,
}

\lstdefinestyle{comment-style}{
	morecomment=[l]//,
	morecomment=[s]{/*}{*/},
	commentstyle={\color{ccomment}\itshape},
}

\lstdefinestyle{string-style}{
	morestring=[b]",%
	morestring=[b]',%
	stringstyle={\color{cstring}},
	showstringspaces=false,%
}

\lstdefinestyle{keyword-style}{
	keywordstyle={\ttfamily\bfseries},
	morekeywords={
		function,
		constructor,
		int,
		bool,
		return,
		returns,
		uint
	},
	morekeywords = [2]{},
	keywordstyle = [2]{\text},
	sensitive=true,
}

\lstdefinestyle{input-encoding}{
	inputencoding=utf8,
	extendedchars=true,
	literate=
	{ℝ}{$\reals$}1%
	{→}{$\rightarrow$}1%
	{α}{$\alpha$}1%
	{β}{$\beta$}1%
	{λ}{$\lambda$}1%
	{θ}{$\theta$}1%
	{ϕ}{$\phi$}1%
}

\lstdefinestyle{escaping}{
	moredelim={**[is][\color{blue}]{\%}{\%}},
	escapechar=|,
	mathescape=true
}

\lstdefinestyle{default-style}{
	basicstyle=\fontencoding{T1}\ttfamily\footnotesize,
	style=numbers,
	style=layout,
	style=comment-style,
	style=string-style,
	style=keyword-style,
	style=input-encoding,
	style=escaping,
	tabsize=2,
	upquote=true
}

\lstdefinelanguage{BASIC}{
	language=C++,
	style=default-style
}[keywords,comments,strings]%

\usepackage[utf8]{inputenc} %
\usepackage[T1]{fontenc}    %
\usepackage{caption}
\usepackage{algpseudocode}
\usepackage{enumerate}
\usepackage{url}            %
\usepackage{booktabs}       %
\usepackage{amsfonts}       %
\usepackage{nicefrac}       %
\usepackage{microtype}      %
\usepackage{xcolor}         %
\usepackage{tikz,booktabs,multirow,enumitem,bm}
\usepackage{colortbl}
\usepackage{tabularx}
\usepackage{subcaption}
\usepackage{wrapfig}
\usepackage{tabulary}
\usepackage{amsmath,amsthm,amsfonts, bbm}
\usepackage{makecell}
\usepackage{rotating}
\usepackage{array}
\usepackage{siunitx}
\usepackage[most]{tcolorbox}
\usepackage{algorithm}
\usepackage{pifont}

\usepackage{amsmath,amsfonts,bm}
\usepackage{amsthm}
\usepackage{amssymb}
\usepackage{mathtools}

\def\1{\bm{1}}

\DeclareMathAlphabet{\mathsfit}{\encodingdefault}{\sfdefault}{m}{sl}
\SetMathAlphabet{\mathsfit}{bold}{\encodingdefault}{\sfdefault}{bx}{n}

\DeclareMathOperator*{\argmax}{arg\,max}

\newcolumntype{x}[2]{S[table-format=#1.#2,table-auto-round]}
\newcolumntype{y}[2]{>{\small} S[table-format=#1.#2,table-auto-round]}

\definecolor{hyperlinkblue}{HTML}{0000AA}
\hypersetup{citecolor=hyperlinkblue} %

\definecolor{red}{HTML}{FF0000}
\definecolor{drop}{HTML}{2596be}
\definecolor{anon}{HTML}{10a37f}
\definecolor{devil}{HTML}{84cf9d}

\lstdefinestyle{mystyle}{
    escapechar=\#,
    breaklines=true,
    basicstyle=\scriptsize\ttfamily,
    numbers=none,
    language={},
    framextopmargin=0pt,
    framexbottommargin=0pt,
    breakindent=0pt,
    showspaces = false,
    keywordstyle=\bfseries,
    showstringspaces=false,
    columns=fullflexible,
    morekeywords={Style, Consistency, Accuracy, Ethics, Score}
}

\newtcblisting{prompt}[2][]{
    arc=3pt, outer arc=3pt,
    width=.9\linewidth,
    left=1mm,
    top=0mm,
    bottom=0mm,
    title=#2, 
    colback=gray!5!white,
    colframe=black!75!black,
    fonttitle=\bfseries,
    listing only, 
    listing options={style=mystyle},
    breakable,
    #1
}

\newtcblisting{response}[2][]{
    arc=3pt, outer arc=3pt,
    width=.9\linewidth,
    left=1mm,
    top=0mm,
    bottom=0mm,
    title=#2, 
    colback=devil!5!white,
    colframe=devil!50!black,
    fonttitle=\bfseries,
    listing only, 
    listing options={style=mystyle,deletekeywords={Style}},
    breakable,
    #1
}

\newtcblisting{inference}[2][]{
    arc=3pt, outer arc=3pt,
    width=.9\linewidth,
    left=1mm,
    top=0mm,
    bottom=0mm,
    title=#2, 
    colback=red!5!white,
    colframe=red,
    fonttitle=\bfseries,
    listing only, 
    listing options={style=mystyle},
    breakable,
    #1
}

\newtcblisting{anonymize}[2][]{
    arc=3pt, outer arc=3pt,
    width=.9\linewidth,
    left=1mm,
    top=0mm,
    bottom=0mm,
    title=#2, 
    colback=anon!5!white,
    colframe=anon,
    fonttitle=\bfseries,
    listing only, 
    listing options={style=mystyle},
    breakable,
    #1
}

\lstdefinelanguage{QA}{
  morekeywords={Prompt,Base,Watermarked},
  sensitive=false,
}

\lstdefinestyle{qaStyle}{
  language=QA,
  basicstyle=\normalfont\small,      %
  keywordstyle=\bfseries\color{teal},%
  stringstyle=\color{black},         %
  numbers=none,                       %
  showstringspaces=false,            %
  frame=single,                      %
  rulecolor=\color{gray!50},         %
  backgroundcolor=\color{gray!10},   %
  columns=fullflexible,              %
  xleftmargin=1em, xrightmargin=1em, %
  aboveskip=1em, belowskip=1em       %
}

\definecolor{outerbg}{HTML}{F6F6F6}      %
\definecolor{outerframe}{HTML}{D0D0D0}   %

\definecolor{modelAback}{HTML}{DD8452}   %
\definecolor{modelAframe}{HTML}{DD8452}  %

\definecolor{modelBback}{HTML}{4C72B0}   %
\definecolor{modelBframe}{HTML}{4C72B0}  %

\newtcolorbox{promptbox}[1][]{%
  breakable,
  colback=white,
  colframe=black!25,
  fonttitle=\bfseries\small,
  coltitle=black!80,
  title=User~Prompt~{#1},
  boxsep=5pt,
  top=4pt,bottom=4pt,
  arc=1mm,
}

\newtcolorbox{modelbox}[3][]{%
  breakable,
  colback=#2!75,
  colframe=#3!85,
  coltitle=#3!10!black,
  fonttitle=\bfseries\small,   %
  title={#1},
  boxsep=5pt,
  top=4pt,bottom=4pt,
  arc=1mm,
}

\renewcommand{\S}{Sec.~}

\usepackage[capitalize,noabbrev]{cleveref}

\crefformat{section}{\S#2#1#3}

\crefrangeformat{section}{\S#3#1#4\crefrangeconjunction\S#5#2#6}

\crefmultiformat{section}{\S#2#1#3}{\crefpairconjunction\S#2#1#3}{\crefmiddleconjunction\S#2#1#3}{\creflastconjunction\S#2#1#3}

\newcommand{\crefrangeconjunction}{--}

\crefname{listing}{Lst.}{listings}
\crefname{line}{Lin.}{Lin.}
\crefname{appendix}{App.}{App.}

\newcommand{\appref}[1]{%
	\ifbool{includeappendix}{\cref{#1}}{the appendix}%
}
\newcommand{\Appref}[1]{%
	\ifbool{includeappendix}{\cref{#1}}{The appendix}%
}

\usepackage{aliascnt}
\theoremstyle{plain}
\newtheorem{theorem}{Theorem}[section]
\newtheorem{proposition}[theorem]{Proposition}
\newaliascnt{lemma}{theorem}
\newtheorem{lemma}[lemma]{Lemma}
\aliascntresetthe{lemma}
\crefname{lemma}{Lemma}{Lemmas}
\Crefname{lemma}{Lemma}{Lemmas}

\theoremstyle{definition}

\theoremstyle{remark}

\newcommand{\OurTitle}{Every Bit, Everywhere, All at Once:\\A Binomial Multibit LLM Watermark}

\title{\OurTitle{}}

\author{%
  Thibaud Gloaguen, Robin Staab, Mark Vero, Martin Vechev \\
  ETH Zurich \\
  \texttt{thibaud.gloaguen@inf.ethz.ch}
}

\begin{document}

\maketitle

\begin{abstract}
	With LLM watermarking already being deployed commercially, practical applications increasingly require \emph{multibit} watermarks that encode more complex payloads, such as user IDs or timestamps, into the generated text. 
In this work, we propose a fundamentally new approach for multibit watermarking: introducing binomial encoding to directly encode every bit of the payload at every token position. 
We complement our approach with a \emph{stateful encoder} that during generation dynamically redirects encoding pressure toward underencoded bits. 
Our evaluation against 8 baselines on up to 64-bit payloads shows that our scheme achieves superior message accuracy and robustness, with the gap to baseline methods widening in more relevant settings (i.e., large payloads and low-distortion regimes). 
At the same time, we challenge prior works’ evaluation metrics, highlighting their lack of practical insights, and introduce \emph{per-bit confidence scoring} as a practically relevant metric for evaluating multibit LLM watermarks.
Our code is available \href{https://github.com/eth-sri/multibit-binomial-wm}{\textcolor{hyperlinkblue}{here}}.
\vspace{-0.1in}

\end{abstract}

\section{Introduction}
\label{sec:intro}

LLM watermarking has emerged as a leading approach for detecting AI-generated content, with major companies and regulators moving toward deployment~\citep{euai}.
Initial watermarking research focused on \emph{zero-bit} watermarks, which answer a single binary question: was this text generated using a given watermark?
While technically sufficient for detection, emerging practical applications demand richer watermarking signals that, \eg allow model providers to embed individual user IDs,  encode timestamps, or include licensing metadata.
These use cases require \emph{multibit} watermarks that encode an $m$-bit payload into the model output and are able to reliably decode it from watermarked text.

\paragraph{Multibit Watermarking}
To enable the inclusion of such multibit information, existing schemes overwhelmingly rely on \emph{position allocation}~\cite{mpac,rsbh,mirrormark}: at each generation step, the watermark selects a single bit position from the payload and encodes only that bit using a chosen zero-bit scheme.
The full message is then recovered by aggregating per-token evidence across many tokens.
Yet, allocated positions can be uneven~\citep{mirrormark}, and encoding only one bit per token has been found to be sub-optimal~\citep{arcmark}, limiting the effectiveness of such schemes in practice (see \cref{sec:evaluation}).
At the same time, alternative, more powerful, approaches that avoid position allocation, \eg Cycle-Shift~\cite{three_bricks} and ArcMark~\cite{arcmark}, are computationally limited to short payloads.
Hence, the question remains whether we can design a more powerful multibit watermarking scheme that effectively encodes longer bitstrings without having to assign individual bits to individual generation steps?

\paragraph{This Work: Multibit Watermarking via Binomial Encoding}
We answer this question affirmatively with a powerful multibit watermarking scheme that encodes large payloads without position allocation (illustrated in \cref{fig:accept}).
Our key insight is to replace \emph{fully encoding individual bits in specific locations} with \emph{partly encoding every bit across every location} via \emph{binomial encoding}.
Given an $m$-bit message, at each generation step, we sample $m$ independent Bernoulli score vectors and flip each according to the corresponding bit value, yielding a \emph{binomial score} that aggregates alignment across all $m$ bits (\cref{fig:accept}, blue shade).
The watermarked distribution then favors tokens whose scores align with the full message (\cref{fig:accept}, middle).
Decoding reduces to a majority vote per bit, with statistical significance provided by two-sided binomial tests (\cref{fig:accept}, right).
We further introduce a \emph{stateful encoder} that, during generation, tracks which bits are already well-encoded and redirects encoding pressure toward bits that are not yet correctly recovered.
Our evaluation (\cref{sec:evaluation}) shows that our approach outperforms prior works in most practical scenarios.

Beyond the method itself, we also revisit how multibit watermarks should be evaluated under practical deployment conditions.
Prior works~\cite{mpac, rsbh, mirrormark} typically report bit accuracy and message accuracy assuming the input text is already known to be watermarked; that is, they answer ``given a watermarked text, how much of the encoded message can we decode?''
They do not address the broader deployment question: given an arbitrary text, can we first detect whether it is watermarked and, if so, reliably decode the embedded message?
Since most texts encountered in practice are not watermarked, ignoring false positives in message recovery falls short of real-world deployment constraints.
Notably, we find in \cref{sec:considerations} that while existing approaches can be post-hoc adapted with \emph{error-detection codes}, this comes at the cost of reduced capacity.
In contrast, binomial watermarks directly support \emph{per-bit confidence scoring}, enabling even statistically sound recovery of sub-messages.

\paragraph{Key Contributions.} Our main contributions are:
\begin{itemize}
    \item We introduce the first multibit LLM watermark that encodes the full message bitstring at every token position via \emph{binomial encoding}, avoiding existing shortcomings (\cref{subsec:multibit_encoding_decoding}).

    \item We propose a \emph{stateful encoder} that dynamically prioritizes underperforming bits during generation, significantly improving bit accuracy without affecting decoding (\cref{subsec:multibit_improved}).

    \item We challenge prior multibit evaluation metrics, highlighting their limited interpretability, and propose per-bit confidence scoring as an alternative metric (\cref{sec:considerations}). 

    \item We provide a comprehensive evaluation against 8 baselines across 3 bitstring lengths (16, 32, 64 bits), showing that our scheme achieves superior bit accuracy, message accuracy, and robustness, with the gap widening at longer payloads and in low-distortion regimes (\cref{sec:evaluation}).
\end{itemize}

\section{Background and Related Work}
\label{sec:related_work}

\begin{figure}
    \centering
    \includegraphics[width=\linewidth]{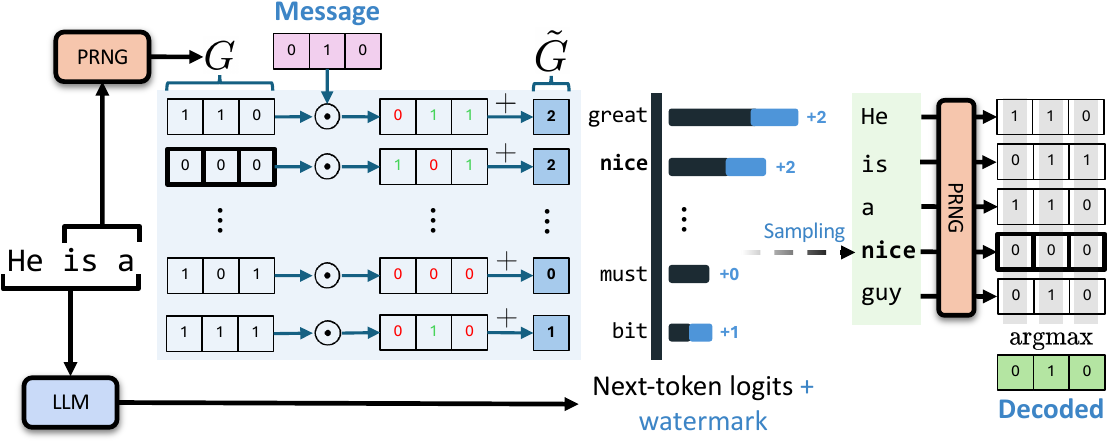}
    \caption{\textbf{Overview of Our Multibit Watermark Algorithm Encoder and Decoder:} Given an $m$-bit bitstring (here \texttt{010}), at each step of the generation process, we use the context to seed a pseudorandom number generator (PRNG) that samples $m$ Bernoulli variables ($G$) for each token in the vocabulary. We encode our bitstring with an \emph{XNOR} operation and sum the encoded Bernoulli variables to obtain binomial token scores ($\tilde{G}$). We then modify the next-token logit distribution with $\tilde{G}$ to obtain our watermarked distribution. For decoding the bitstring from a text, we use the same PRNG to retrieve the $m$ Bernoulli variables for every token in the text. We then decode the bitstring by taking the majority bit at each bit position, using $m$ binomial tests to obtain per-bit confidence.}
    \label{fig:accept}
    \vspace{-0.1in}
\end{figure}

In this section, we introduce relevant related work in (multibit) LLM watermarking.

\paragraph{LLM Watermarks}
The goal of (zero-bit) LLM watermarks is to insert a detectable signal into text generated by an LLM.
At each step of the generation process, the previously generated context and a \emph{private key} seed a pseudorandom number generator, which assigns a \emph{score} to each token in the vocabulary.
Then, the next-token distribution produced by the LLM is biased to favor tokens with higher scores.
For detection, given a sequence of tokens (text), we aggregate the individual token scores.
We call this aggregation the detection \emph{statistic}. 
For instance, with Red-Green watermarks~\citep{kgw1}, this statistic is the number of green tokens.
If the text was generated without knowledge of the scores, the statistic is random.
However, when using the watermark, the statistic is biased toward higher values.
Importantly, for deploying detection at internet-scale, the statistic should be model-free, \ie it should not require access to the LLM.

Popular approaches commonly build on \citet{kgw1}, which biases the logits of the next-token distribution according to binary scores.
\citet{synthid} adapts this by proposing a tournament sampling strategy, where tokens compete against each other and the highest-scoring tokens are selected.
\citet{aar} leverages the Gumbel-max trick to deterministically sample the next token.
Lastly, \citet{dipmark,mcmark} reweight the next-token distribution according to the ranking of the token scores.
Each watermark approach strikes a trade-off between the watermark impact on quality and its detectability.
Notably, most~\citep{synthid, aar, kth, dipmark, mcmark} are distortion-free: in expectation over the private key they do not modify the LLM next-token distribution. 

\paragraph{Multibit LLM Watermarks}
Multibit LLM watermarks inject a bitstring (\eg user ID, timestamp) into the generated text that can later be decoded.
\citet{three_bricks} (\emph{Cycle-Shift}) proposes treating each bitstring as a private key and then using any zero-bit watermarking algorithm.
For decoding, they compute the statistic with each possible key and decode to the key with the highest statistic.
While effective for small messages, this method does not scale to longer ones as decoding time increases exponentially with the message length.
\citet{mpac} uses position allocation (\emph{MPAC}) to encode the bitstring with any zero-bit watermarking algorithm: at each generation step, they use the previously generated context to pseudorandomly select a position within the bitstring.
Given the bit value at this position, they modify the score by taking its complement (defined as $F^{-1}(1 - F(\text{score}))$, where $F$ is the score CDF and $F^{-1}$ its quantile function) if the bit value is $1$.
For decoding, they accumulate the scores and complementary scores respectively for each position and compute a (position-wise) statistic with either.
The decoded bit value for a position is set to $0$ if its statistic computed with the scores is higher than with the complementary scores, and $1$ otherwise.
\citet{mpac} also suggests extending MPAC with Cycle-Shift: for each position, a sub-bitstring can be encoded using Cycle-Shift instead of complementary scores.
This method also does not scale to longer bitstrings because decoding time increases exponentially with the message length.

Various prior works build on MPAC and Cycle-Shift.
\citet{rsbh} (\emph{RSBH}) combines a balanced position allocation algorithm with Red-Green watermarks~\citep{kgw1} and error-correcting codes to increase robustness.
\citet{stealthink} (\emph{StealthInk}) combines MPAC with DiPMark~\citep{dipmark}, and \citet{bimark} (\emph{BiMark}) combines MPAC with an unbiased reweighting strategy to minimize the impact of the watermark on quality.
More recently, \citet{mirrormark} (\emph{MirrorMark}) proposes CABS, improving MPAC to better balance the selected position, and combines it with SynthID-text~\citep{synthid}.
\citet{mc2mark} (\emph{MC2Mark}) expands McMark~\citep{mcmark} to the multibit setting.
Lastly, \citet{arcmark} (\emph{ArcMark}) utilizes random linear channel codes and optimal transport on a unit circle to embed the watermark, ensuring that each generated token contains some information about the entire message.
However, for ArcMark, similarly to Cycle-Shift, the decoding time also scales exponentially with the message length~\citep{arcmark}.

Importantly, all approaches either only encode a sub-bitstring per token and must rely on position allocation to encode longer bitstrings or are intractable for longer messages~\citep{three_bricks,arcmark}.
In this work, we introduce the first multibit watermark that can effectively encode any bitstring length without position allocation, and, crucially, is more powerful than prior works, as we show in \cref{sec:evaluation}. %

\section{Our Method}
\label{sec:method}

Next, we introduce our multibit LLM watermarking algorithm based on \emph{binomial encoding}.
First, we explain how our approach encodes and decodes payloads into generated token sequences~(\cref{subsec:multibit_encoding_decoding}).
Then, we introduce an optional stateful encoder that significantly improves the bit accuracy~(\cref{subsec:multibit_improved}).
Lastly, we combine all our elements into an end-to-end watermarking algorithm~(\cref{subsec:full_watermark}).

\subsection{Encoding and Decoding the Bitstring}
\label{subsec:multibit_encoding_decoding}

\paragraph{Notation}
We adopt the notation from~\citet{gloaguen2026unifiedframeworkllmwatermarks}.
At each token position $t$, given the current context $\omega_{<t}\in\Sigma^{t-1}$, let $p_t \in \Delta(\Sigma)$ be the next-token probability distribution given by the LLM.
We generate, using the context $\omega_{<t}$ and the private key, (pseudorandom) watermark scores $G_t \in \mathbb{R}^{|\Sigma|}$ (\ie one score per candidate token), and we denote by $q(G_t,p_t) \in \Delta(\Sigma)$ the corresponding watermarked next-token probability distribution used for sampling at position $t$.

We note that as in \citep{gloaguen2026unifiedframeworkllmwatermarks}, $q(G_t,p_t) \in \Delta(\Sigma)$ generalizes a range of watermarking schemes, such as SynthID or Red-Green watermarks.
In particular, this means that our multibit en-/decoding approach can be easily adapted to use most zero-bit watermarks.
For ease of illustration, both this section and \cref{fig:accept} present our method using Red-Green watermarks, whereas we provide details on the exact $q(G_t,p_t)$ in \cref{subsec:full_watermark}.
With Red-Green watermarks, we have $G_t \in \{0,1\}^{|\Sigma|}$ (where $0$ are red tokens and $1$ green tokens) and, more precisely, each entry is an i.i.d.\ Bernoulli (pseudo)-random variable with probability $0.5$. The watermarked next-token probability distribution is given by
\begin{equation} \label{eq:red_green_singlebit}
    q(G_t,p_t) \propto p_t \exp(\delta G_t),
\end{equation}
where $\delta > 0$ is a hyperparameter that balances the watermark detectability-quality trade-off.

\paragraph{Binomial Encoding}
Let $M \in \{0,1\}^m$ be the $m$-bit bitstring that we want to encode in the generated text.
Unlike most prior works~\citep{mpac}, we propose encoding the full bitstring $M$ at every position $t$.
Our key insight is to encode the $m$ bits directly into the scores $G_t$.
To do so, at each position $t$ we sample $m$ independent Bernoulli(0.5) score vectors $G_t^1,\dots,G_t^m \in \{0,1\}^{|\Sigma|}$ (i.e., each vector defines a split of the entire vocabulary).
The score vector $G_t^i$ is used to encode the $i$-th bit of our bitstring.
For each $i \in \{1,\dots,m\}$, we encode the $i$-th bit $M_i$ using complementary scores,
\begin{equation} \label{eq:bernoulli_transform}
    \tilde{G}^i = \left\{
    \begin{array}{ll}
        G^i     & \mbox{if } M_i = 1, \\
        1 - G^i & \mbox{otherwise.}
    \end{array}
    \right.
\end{equation}
At position $t$, we apply \cref{eq:bernoulli_transform} to each $G_t^i$ and obtain complementary score vectors $\tilde{G}_t^1,\dots,\tilde{G}_t^m \in \{0,1\}^{|\Sigma|}$, allowing us to construct the final score vector $\tilde{G}_t \coloneq \sum_{i=1}^m \tilde{G}_t^i \in \{0,\dots,m\}^{|\Sigma|}$.
Given a token $u \in \Sigma$, $\tilde{G}_t(u)$ indicates how many components $\tilde{G}_t^i(u)$ align with the original bitstring $M$ (\ie it gives the bit accuracy when decoding the bitstring from token $u$ at position $t$).
Sampling tokens with larger $(\tilde{G}_t)_u$, \ie with higher scores, therefore increases the expected bit accuracy.
Hence, we can compute $q(\tilde{G}_t,p_t)$, \ie the watermarked next-token probability distribution that favors high-scoring tokens, which yields the embedding step of our watermark at position $t$. 
We illustrate this encoding process in \cref{fig:accept} (blue shade).

With the Red-Green watermark, the multibit watermarked distribution is then given by
\begin{equation}
    q(\tilde{G}_t,p_t) \propto p_t \exp(\delta \tilde{G}_t).
\end{equation}
This probability distribution biases sampling towards tokens with larger $\tilde{G}_t$.
The only difference to \cref{eq:red_green_singlebit} is that, instead of red/green Bernoulli scores, we now use message-bitstring-aligned binomial scores. We illustrate this adjusted Red-Green watermark in \cref{fig:accept} (middle).

\paragraph{Decoding from a Text}
Let $\omega \in \Sigma^*$ be a sequence of tokens.
Decoding aims to extract a bitstring $\hat{M}$ from $\omega$ (assuming our encoding), without relying on the original bitstring $M$ or the LLM next-token probability distributions.
For each token position $t$ (where $\omega_t \in \Sigma$ denotes the token in the sequence), we recompute the same pseudorandom Bernoulli scores $G_t^1(\omega_t),\dots,G_t^m(\omega_t) \in \{0,1\}$ as during encoding (using $\omega_{<t}$ and the private key).
By \cref{eq:bernoulli_transform}, we know that for all $i$, $G_t^i(\omega_t)$ tends to align with the value of $M_i$, \ie if $G_t^i(\omega_t) = 0$ then $M_i$ is likely $0$, and vice versa. Hence, the decoded message at position $t$ is $\hat{M}^t = [G_t^1(\omega_t),\dots,G_t^m(\omega_t)] \in \{0,1\}^m$.
For the final message, we treat each token in the sequence independently and aggregate the per-token evidence using majority vote,
\begin{equation}
    \hat{M} = \text{round}\left( \frac{1}{|\omega|} \sum_{t=1}^{|\omega|} \hat{M}^t \right).
\end{equation}
Importantly, our decoding method also provides statistical significance for each decoded bit.
Under the null hypothesis (\ie assuming the text is generated without the watermark), we can assume that for each $i \in \{1,\dots,m\}$ the sequence $(\hat{M}^t_i)_{t=1}^{|\omega|}$ behaves as i.i.d.\ Bernoulli$(0.5)$ random variables. We can therefore use a two-sided Binomial test to assess whether the decoded bit $\hat{M}_i$ is significant.
We explain how to leverage this property in \cref{sec:considerations}, and illustrate the decoding in \cref{fig:accept} (right).

\paragraph{Zero-Bit Detection}
Our multibit watermark can also be used as a zero-bit watermark to decide whether the text is watermarked using our binomial approach.
While we could directly aggregate the statistical tests from the multibit decoding algorithm, we propose a more powerful test in \cref{app:zerobit_detection}.

This concludes our baseline multibit encoding and decoding algorithm. We provide further practical details in \cref{app:experimental_details} and a full algorithmic description in \cref{app:algorithms}.

\subsection{Improving the Bit Accuracy with a Stateful Encoder}
\label{subsec:multibit_improved}

Next, we increase the bit accuracy of our multibit watermark by improving the encoding scheme from \cref{subsec:multibit_encoding_decoding}, while keeping the decoding strategy fixed.
Our idea is that, when sampling a token at position $t$, we can use our knowledge of the partially generated sequence to assess which bits are already well decoded and which are not.
Then, instead of favoring all bits independently via $\tilde{G}_t$, we weight the importance of each bit, favoring bits that are not yet correctly decoded.
Specifically, to quantify the importance of each bit, we answer for each bit the following question: what is the expected accuracy of the watermarking algorithm if at this step this bit is correctly encoded?

\paragraph{Expected (Worst-Case) Bit Accuracy}
Let $\omega_{\le t} \in \Sigma^t$ be the partially generated sequence.
We introduce the current unnormalized decoded message at position $t$,
\begin{equation} \label{eq:dit_definition}
    \forall i \in \{1,\dots,m\}, \quad d_t^i = 2 \sum_{k=1}^{t} \tilde{G}_k^i(\omega_k) - t,
\end{equation}
where $\tilde{G}_k^i(\omega_k)\in\{0,1\}$ is the complementary score (for bit $i$) of the realized token $\omega_k$ at position $k$.
By design, $\tilde{G}_k^i(\omega_k) = 1$ if the $k$-th token encodes the $i$-th bit correctly. 
Because we use majority voting~(\cref{subsec:multibit_encoding_decoding}), the $i$-th bit is currently correctly decoded if and only if $d_t^i \ge 0$ (ignoring ties).

Let $T \in \mathbb{N}$ denote the total number of generated tokens.
Assuming that, after step $t$, the remaining tokens are generated independently of the watermark (\ie a worst-case scenario), the expected accuracy of the $i$-th bit is given by $P[d_T^i \ge 0 \mid d_t^i]$.
The total expected bit accuracy is therefore the sum across all bits of $P[d_T^i \ge 0 \mid d_t^i]$.
Because the goal of our watermark encoding algorithm is to maximize the expected bit accuracy, instead of using $\tilde{G}$ as scores, we should directly use the expected bit accuracy with the closed-form expression from \cref{lemma:expected_bit_acc_under_null}.
We defer its proof to \cref{app:proofs}.

\newcommand{\ExpectedBitAccUnderNullStatement}[1]{%
    Let $t \in \mathbb{N}$.
    Let $\omega_{\le t} \in \Sigma^t$ be a partially generated sequence, and let $d_t^i$ denote the signed decoding statistic for bit $i$, as defined in \cref{eq:dit_definition}.
    Assuming that future tokens are generated independently of the watermark, we have for all $i \in \{1,\dots,m\}$ and $T > t$
    \begin{equation}
        #1
        P[d_T^i \ge 0 \mid d_t^i] \approx \Phi\left(\frac{d_t^i}{\sqrt{T-t}}\right),
    \end{equation}
    where $\Phi$ is the standard normal CDF.
}

\begin{lemma} \label{lemma:expected_bit_acc_under_null}
    \ExpectedBitAccUnderNullStatement{\label{eq:expected_bit_accuracy_formula}}
\end{lemma}

In practice, $T$ is not known in advance.
We ablate this parameter in~\cref{app:additional_experiments:ablation_remaining_bits} and find that averaging $T$ over $\mathcal{T} \coloneq \{200,300,500,1000,2000\}$ yields a significantly higher bit accuracy compared to using $\tilde{G}$.
Hence, at time $t$ for each candidate token $u \in \Sigma$, given the current pseudorandom scores $G_t^1(u),\dots,G_t^m(u)$ (and the corresponding complementary scores $\tilde{G}_t^i(u)$), we can use the following watermarking scores instead of $(\tilde{G}_t)_u$,
\begin{equation} \label{eq:stateful_scores}
    \tilde{G}'_t(u) \coloneq \sum_{i=1}^{m} \frac{1}{|\mathcal{T}|} \sum_{T \in \mathcal{T}} \Phi\left( \frac{1}{\sqrt{T}} (d_{t-1}^i + (2\tilde{G}_t^i(u)-1)) \right).
\end{equation}

\subsection{Our Watermark Algorithm}
\label{subsec:full_watermark}

We now present our end-to-end watermarking algorithm using the multibit encoding and decoding strategy described in~\cref{subsec:multibit_encoding_decoding} and~\cref{subsec:multibit_improved}, and the Soft PPL scheme from \citet{gloaguen2026unifiedframeworkllmwatermarks}.
We select the Soft PPL scheme as it provides, in its zero-bit variant, a high detectability-quality trade-off~\citep{gloaguen2026unifiedframeworkllmwatermarks}.
In \cref{app:synthid_multibit}, we alternatively propose using the SynthID scheme from~\citet{synthid} to build a distortion-free multibit watermark based on our encoding and decoding strategy from~\cref{subsec:multibit_encoding_decoding}.

At each token position $t$, let $\omega_{<t} \in \Sigma^{t-1}$ be the currently generated sequence.
Given $\omega_{<t}$, we use our LLM to generate a next-token probability distribution $p_t\in\Delta(\Sigma)$.
At the same time, using the $k$ previously generated tokens as context (\ie $\omega_{t-k:t-1}$), the watermark private key, and our fixed $m$-bit message $M \in \{0,1\}^m$, we pseudorandomly sample binary score vectors $G_t^1,\dots,G_t^m \in \{0,1\}^{|\Sigma|}$, from which we derive the stateful score $\tilde{G}'_t(u)$ using \cref{eq:stateful_scores}.
We then compute the watermarked next-token distribution using the Soft PPL watermarking scheme from \citet{gloaguen2026unifiedframeworkllmwatermarks},
\begin{equation} \label{eq:soft_ppl_wm}
    \forall u \in \Sigma, \quad q(\tilde{G}'_t,p_t)_u = \mathbf{1}\left\{u = \argmax_{v\in\Sigma}\left(\tilde{G}'_t(v) + \lambda \log (p_t)_v\right)\right\},
\end{equation}
where $\lambda > 0$ is chosen such that $\mathbb{E}_{G}\left[ q(\tilde{G}'_t,p_t)\cdot \log p_t \right] = p_t \cdot \log p_t - \varepsilon$, and $\varepsilon \ge 0$ is a scheme parameter.
Alternatively, we can directly use $\lambda > 0$ as the scheme parameter, yielding similar performance at a lower runtime (\cref{app:unconstrained_soft_ppl,app:additional_experiments:runtime}).
For detection, we use the algorithm from \cref{subsec:multibit_encoding_decoding}.

\section{Considerations for Multibit Watermarks}
\label{sec:considerations}

In this section, we take a step back from watermark design and instead consider how practical multibit watermarks should be evaluated.
In particular, we find that, inadvertently, most prior works evaluate the performance of multibit watermarks \emph{assuming that the text is already known to be watermarked}.

\paragraph{Evaluation of Zero-Bit Watermarks}
The goal of zero-bit watermarks is to decide whether a given input text was generated by our (watermarked) LLM or not.
Most prior works~\citep{kgw1,kth} rely on statistical testing and report the True Positive Rate at a given False Positive Rate (TPR@$\alpha$\%FPR) as a metric for watermark performance.
With $\text{TPR@$\alpha$\%FPR} \coloneq x \in [0,100]$, one knows that when running the watermark detection on \emph{any text}, (i) if the text was generated by the watermarked model, the detector detects the watermark $x$\% of the time and (ii) if the text was not generated by the watermarked model, it wrongly detects the watermark $\alpha$\% of the time.
Yet, reported metrics in prior multibit watermark evaluations usually only measure the first aspect: they implicitly assume that the multibit decoder is always used on \emph{watermarked text} (\ie they always assume a message is embedded in the text).

\paragraph{Evaluation of Multibit Watermarks in Prior Works}
In particular, most prior works~\citep{mpac,rsbh,mirrormark} report the bit accuracy (or message accuracy) of their watermarks as measured \emph{on watermarked samples} (\ie they sample watermarked text with an encoded message $M$, and on these samples compute the decoded message $\hat{M}$).
These metrics answer the following question: \emph{given a text generated by the watermarked LLM (\ie assuming there is an embedded message), what is the probability that a given bit (respectively, the full bitstring) is correctly decoded?}
Yet, given a text, prior methods would also decode a (random) message from text, even if none had been encoded.
Crucially, this means that when decoding the watermark from \emph{any text}, these metrics provide no information about whether the decoded message actually exists.
Hence, we argue that, unlike prior metrics, evaluations of multibit watermarks should measure their capability to answer two questions: (i) is there a message embedded in a given text, and (ii) if so, can it be correctly decoded?

\paragraph{Requirements for Multibit Watermarks}
More formally (and inspired by the zero-bit case), we should report metrics that answer: \emph{given any text}, what is the probability that a given bit (respectively, the full bitstring) is correctly decoded \emph{with $(1-\alpha)$ confidence that the message exists?}
We introduce a new metric to answer this question: bit accuracy at $\alpha$\% FPR (BA@$\alpha$\%FPR) (where, as in the zero-bit case, a false positive is wrongly detecting a non-existing message).
It cannot be derived from TPR and bit accuracy (\cref{app:theory:ba_at_fpr}), and we adapt several prior multibit watermarks to it in \cref{sec:evaluation:bit_acc_confidence}.

Alternatively, using error-detection codes with the encoded message $M$, we can use the message accuracy metric from prior work to answer the question: \emph{given any text, what is the probability that the decoded message exists?}
If this metric is $x \in (0,100)$, one correctly detects the existence of the message at least $x$\% of the time (and most error-detection codes have a low probability of incorrectly detecting the existence), and correctly decodes the message $x$\% of the time.
Yet, this approach requires additional bits to encode a message of length $m$, it provides no information on how often a valid message is wrong, and does not necessarily provide information on which bits are incorrect.
We compare the message accuracy of our approach against all baselines in \cref{sec:evaluation:message_accuracy}.

With our scheme, we find that we can trade off both metrics.
Using the baseline scores $\tilde{G}$ from \cref{subsec:multibit_encoding_decoding}, we obtain higher BA@$\alpha$\%FPR but lower message accuracy, while with the stateful encoder from \cref{subsec:multibit_improved}, we have higher message accuracy but lower BA@$\alpha$\%FPR.

\section{Evaluation}
\label{sec:evaluation}

\begin{figure}
    \centering
    \resizebox{\textwidth}{!}{\input{figures/main/comparison_mer_main.pgf}}
    \caption{\textbf{Comparison of the Message Accuracy-Quality Trade-Off:} We compare the trade-off between message accuracy and text quality ($\log \text{PPL}$) for various multibit watermarks and three bitstring lengths ($16$, $32$, and $64$ bits).
        Each sample contains between $250$ and $350$ tokens and is generated by \llama{} using prompts from \elifive{}.}
    \label{fig:main_comparison_messacc}
    \vspace{-2em}
\end{figure}

\begin{wrapfigure}{r}{0.5\textwidth}
    \centering
    \vspace{-1em}
    \resizebox{\linewidth}{!}{\input{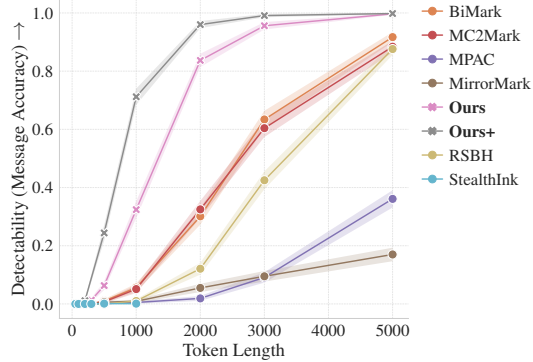}}
    \caption{\textbf{Message Accuracy with Respect to the Number of Tokens:} We compare how the message accuracy scales with the number of generated tokens.
        The hyperparameters of each watermarking scheme have been selected to have a similar impact on quality in the low-distortion regime, and all schemes use a bit string of $32$ bits.}
    \label{fig:token_length_messac}
    \vspace{-2em}
\end{wrapfigure}

In this section, we evaluate the performance of our proposed multibit watermark,
first with message accuracy~(\cref{sec:evaluation:message_accuracy}) for a fair evaluation against all baselines, and then as suggested in \cref{sec:considerations}
against relevant baselines using the bit accuracy at 1\% FPR~(\cref{sec:evaluation:bit_acc_confidence}).
We also measure the robustness of our watermark against text modifications~(\cref{sec:evaluation:robustness}).
In \cref{app:additional_experiments:statistical_metrics} we evaluate the zero-bit detection performance of our scheme, and in \cref{app:additional_experiments:bit_accuracy}, following prior works, we also compare our scheme against all baselines using bit accuracy.
Lastly, in \cref{app:additional_experiments:position_allocation} we show that combining position allocation with our approach does not improve performance.

\paragraph{Experimental Setup}
Unless specified otherwise, we use \llama{} with temperature $0.7$ and top-k set to $50$. For each watermark configuration we generate replies of between $250$ and $350$ tokens to $1000$ prompts from the \elifive{} dataset~\citep{eli5}.
For the baselines, we use the default parameters recommended by prior work (\cref{app:experimental_details:baselines_watermark_config}).
For our schemes, we denote by \emph{Ours} the version using $\tilde{G}$ and \emph{Ours+} the version using $\tilde{G}'$, and discuss additional parameters in \cref{app:experimental_details:our_watermark_config}.
We sample the bitstring uniformly at random for each completion and measure perplexity with \qwen{}.
\cref{app:experimental_details} provides further hyperparameter details, and in \cref{app:additional_experiments:mistral,app:additional_experiments:scenarios}, we evaluate on another model (\mistral{}) and with other sampling parameters and datasets.

\subsection{Message Accuracy}
\label{sec:evaluation:message_accuracy}

\paragraph{Message Accuracy}
Message accuracy measures, given a watermarked text, the probability that the decoded message is correct (\ie that no single bit is incorrectly decoded).
As explained in \cref{sec:considerations}, in practice, the message accuracy metric must be interpreted with error-detection codes in mind.
If a scheme reaches 100\% message accuracy for an $m$-bit bitstring, it effectively encodes $m$ minus the number of error-detection bits.
For the most widely used error-detection family, cyclic redundancy checks~\citep{crc_error_detection}, $k$ bits are used to achieve a probability of incorrectly detecting a message of $\approx 2^{-k}$.

\paragraph{Improved Message Accuracy}
In \cref{fig:main_comparison_messacc}, we show the message accuracy--quality trade-off of different watermarking algorithms for varying bitstring lengths, using perplexity as a proxy for quality.
For larger, more practical payloads ($32$ and $64$ bits), our watermark (Ours+) outperforms prior work (see \cref{app:additional_experiments:payload8} for $8$ bits).
In particular, we find that the stateful encoder~(\cref{subsec:multibit_improved}) significantly improves the message accuracy, as intended, compared to the stateless version~(\cref{subsec:multibit_encoding_decoding}).
The RSBH baseline uses an error-correction code to improve message accuracy, which explains why it outperforms MPAC.
As mentioned in \citet{rsbh}, this approach is independent of the multibit algorithm and could also be used with our schemes to improve their message accuracy ad hoc.
In \cref{fig:token_length_messac}, we show how message accuracy scales with the number of tokens in the low-distortion regime with a 32-bit payload.
We find that our scheme scales significantly more favorably with the number of tokens compared to all tested baselines.
In \cref{app:additional_experiments:quality_metrics}, we show that our watermark also outperforms prior work when using accuracy on four benchmarks as a proxy for text quality.

\subsection{Bit Accuracy at 1\% FPR}
\label{sec:evaluation:bit_acc_confidence}

\begin{figure}
    \centering
    \begingroup
    \ifdefined\mathdefault\else\def\mathdefault#1{#1}\fi
    \resizebox{\textwidth}{!}{\input{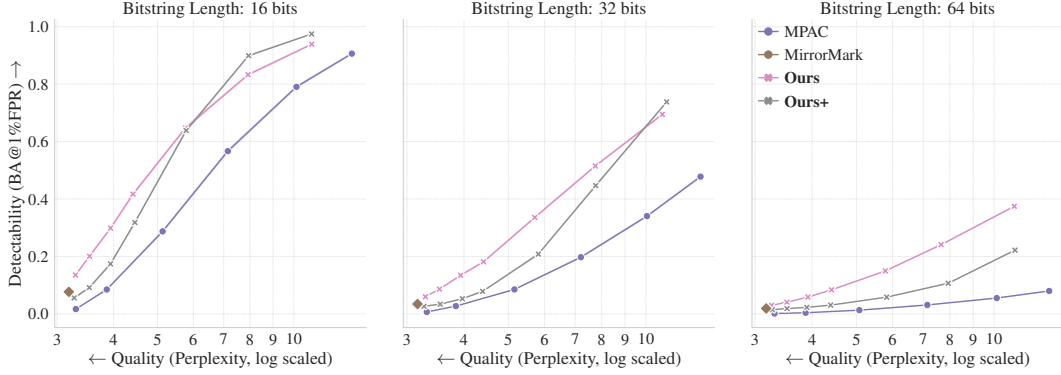}}
    \endgroup
    \caption{\textbf{Comparison of the Confidence Bit Accuracy-Quality Trade-Off:} We compare the trade-off between confidence bit accuracy (BA@1\%FPR) and text quality ($\log \text{PPL}$) across several multibit watermarks and three bitstring lengths ($16$, $32$, and $64$ bits). 
    Each sample contains between $250$ and $350$ tokens and is generated by \llama{} using prompts from \elifive{}.}
    \label{fig:bat_at_1_fpr_comparison}
\end{figure}

\begin{wrapfigure}{r}{0.5\textwidth}
    \centering
    \vspace{-1.5em}
    \resizebox{\linewidth}{!}{\input{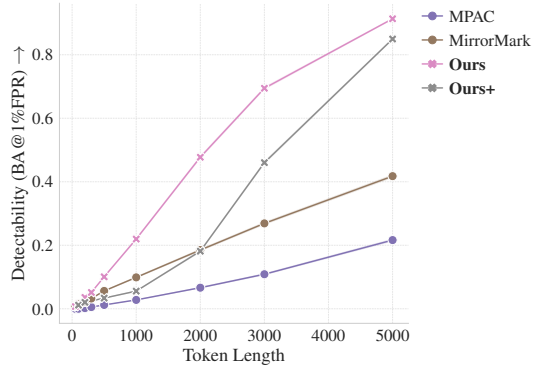}}
    \caption{\textbf{Bit Accuracy at 1\% FPR with Respect to the Number of Tokens:} We compare how BA@1\%FPR scales with the number of generated tokens.
        The hyperparameters of each watermarking scheme are selected to have a similar impact on quality in the low-distortion regime, and all schemes use a bit string of $32$ bits.}
    \label{fig:token_length_bitacc_confidence}
    \vspace{-2em}
\end{wrapfigure}

\paragraph{Per-Bit Confidence}
As explained in \cref{sec:considerations}, we introduce a new, more realistic metric to measure multibit watermark performance: bit accuracy at 1\% FPR (BA@$1\%$).
For a given bit, the BA@$1\%$ FPR is equal to 1 if the decoded bit is correct and its p-value is below $0.01$.
Hence, if BA@$1\%$ FPR is $x \in (0,100)$\%, it means that (i) if the text was generated by the watermarked model, then the detector detects the watermark $x$\% of the time, and (ii) if the text was not generated by the watermarked model, then it wrongly decodes the bit (instead of detecting that there is no encoded bit) at most 1\% of the time.
To compute this metric, we need multibit watermarking schemes that provide per-bit p-values.
As stated in \cref{subsec:multibit_encoding_decoding}, our method provides per-bit p-values.
Additionally, we find that two baselines (MPAC and MirrorMark) can also be used to provide per-bit p-values.
For both MPAC and MirrorMark, by using $1$ bit per position, this is equivalent to using $m$ separate zero-bit watermarks on each position.
To get per-bit p-values, we then use the corresponding two-sided tests (\eg a binomial test for Red-Green watermarks).

\paragraph{Results}
In \cref{fig:bat_at_1_fpr_comparison}, we show the BA@$1\%$ FPR-quality trade-off for our watermark and the two baselines, using perplexity as a proxy for quality.
We find that our stateless approach (\cref{subsec:multibit_encoding_decoding}) yields a significantly higher BA@$1\%$ FPR than the baselines.
At the same time, the stateful version of our watermark (Ours+ from \cref{subsec:multibit_improved}) underperforms the stateless version (Ours).
As mentioned in \cref{sec:considerations}, this outcome is expected: the objective of the stateful encoder is to maximize the expected bit accuracy independently of the confidence.
Additionally, \cref{fig:token_length_bitacc_confidence} shows how BA@$1\%$ scales with token length, using 32-bit bitstrings in the low-distortion regime.
We find that, similarly to \cref{sec:evaluation:message_accuracy}, BA@$1\%$ for our scheme, using the stateless encoder, scales significantly faster than all baselines.

\vspace{-1mm}
\subsection{Robustness to Modifications}
\label{sec:evaluation:robustness}

\paragraph{Experimental Setup}
We evaluate the robustness of our scheme (Ours+) and of the baselines in the low-distortion regime with a $32$-bit bitstring.
As text modifications, we consider word deletion and synonym substitution as in~\citet{markllm}, and paraphrasing using \textsc{GPT-5-mini} as the paraphraser.
For the detectability metric, we report the bit accuracy instead of message accuracy because for all schemes (including ours), the message accuracy after mild text modification is almost null.

\paragraph{Robustness Evaluation}
\cref{tab:robustness_eval_main} shows the bit accuracy for different watermarks and text modifications.
We find that the robustness of all tested watermarks is limited, with bit accuracy sharply dropping beyond $10$\% of the words edited, and becoming almost random (\ie equal to $0.5$) under paraphrasing.
For message accuracy, it is already null at $10\%$ of the words edited.
These results suggest that robust multibit watermarking (in the low-distortion regime and at standard zero-bit evaluation lengths) is still an open problem.
Nonetheless, for all tested text modifications, our watermark outperforms all prior baselines, also when compared at matched quality (\cref{app:additional_experiments:robustness_quality}).

  \begin{table*}[t]
      \centering
      \caption{\textbf{Robustness Evaluation:} Bit accuracy for different schemes with various text modifications (\eg word deletion, synonym substitution, and paraphrasing) averaged over $1000$ samples.
      The percentage corresponds to the percentage of words deleted/substituted.
      Each sample is generated with \llama{} using prompts from \elifive{} and using a $32$-bit bitstring.
      The highest accuracy in each row is \textbf{bolded} for readability.}
      \label{tab:robustness_eval_main}
    
      \renewcommand{\arraystretch}{1.2}
      \newcommand{\skiplen}{0.000001\linewidth} 
      \newcommand{\rlen}{0.01\linewidth} 
      \resizebox{\linewidth}{!}{%
      \begingroup 
      \setlength{\tabcolsep}{5pt} %
\begin{tabular}{lccccccccccc}
\toprule
 & \multicolumn{5}{c}{Word Deletion} & \multicolumn{5}{c}{Synonym Substitution} &  \\
\cmidrule(lr){2-6}\cmidrule(lr){7-11}
Watermark & 10\% & 20\% & 30\% & 40\% & 50\% & 10\% & 20\% & 30\% & 40\% & 50\% & \multicolumn{1}{c}{Paraphrasing} \\
\midrule
BiMark & 0.664 & 0.603 & 0.564 & 0.539 & 0.520 & 0.666 & 0.616 & 0.580 & 0.556 & 0.536 & 0.504 \\
MC2Mark & 0.664 & 0.608 & 0.571 & 0.540 & 0.519 & 0.671 & 0.619 & 0.578 & 0.551 & 0.536 & 0.509 \\
MPAC & 0.609 & 0.572 & 0.544 & 0.529 & 0.524 & 0.612 & 0.575 & 0.555 & 0.532 & 0.524 & 0.505 \\
MirrorMark & 0.539 & 0.521 & 0.505 & 0.506 & 0.500 & 0.544 & 0.523 & 0.509 & 0.511 & 0.505 & 0.498 \\
RSBH & 0.511 & 0.508 & 0.501 & 0.500 & 0.502 & 0.509 & 0.506 & 0.503 & 0.499 & 0.501 & 0.501 \\
StealthInk & 0.606 & 0.569 & 0.546 & 0.525 & 0.515 & 0.609 & 0.572 & 0.549 & 0.537 & 0.526 & 0.507 \\
Ours+ & \textbf{0.707} & \textbf{0.632} & \textbf{0.579} & \textbf{0.541} & \textbf{0.525} & \textbf{0.707} & \textbf{0.642} & \textbf{0.591} & \textbf{0.562} & \textbf{0.543} & \textbf{0.509} \\
\bottomrule
\end{tabular}
      \endgroup
      }
      \vspace{-1.5em}
    \end{table*}

\vspace{-1mm}
\section{Conclusion and Limitations}
\label{sec:conclusion}
\vspace{-1mm}

Our work introduces the first multibit watermarking algorithms that can be applied to large payloads without relying on position allocation, opening a new direction for multibit watermarking research.
Our scheme outperforms all tested baselines in most realistic scenarios.
We also highlight that prior metrics used to evaluate multibit LLM watermarks fall short of practical considerations, and introduce a new metric, BA@1\%FPR, for a more practical evaluation of multibit watermarks.

\paragraph{Limitations}
We show that the robustness of multibit LLM watermarks is still an ongoing challenge.
While our approach outperforms all evaluated baselines, it remains limited in practical scenarios.
Although we show empirically that our approach outperforms position allocation-based baselines, we only prove that encoding every bit at every position never loses power compared to position allocation (\cref{app:theory:position_allocation}).
Proving a strict advantage in general is an interesting direction for future work.

\message{^^JLASTBODYPAGE \thepage^^J}

\clearpage

\bibliographystyle{unsrtnat}
\bibliography{references}
\vfill
\clearpage

\message{^^JLASTREFERENCESPAGE \thepage^^J}

\ifincludeappendixx
	\newpage
	\appendix
	\onecolumn
	\crefalias{section}{appendix}
	\crefalias{subsection}{appendix}
	\section{Experimental Details}
\label{app:experimental_details}

In this section, we describe precisely the experimental setup used in our evaluation (\cref{sec:evaluation}), with the exact inference parameters used (\cref{app:experimental_details:inference_settings}), and the watermark configurations (\cref{app:experimental_details:baselines_watermark_config,app:experimental_details:our_watermark_config}).

\subsection{Inference Setting}
\label{app:experimental_details:inference_settings}

\paragraph{Models}
We use the instruction-tuned versions of \llama{} in \cref{sec:evaluation} and \mistral{} in \cref{app:additional_experiments:mistral}, with temperature $0.7$ and top-k set to $50$.
In \cref{app:additional_experiments:scenarios}, we additionally evaluate sampling with temperature $1$ (and top-k $50$) and with top-k set to $25$ (and temperature $0.7$).
We use both models in bf16 precision.
For inference, we use vLLM~\citep{vllm} on a single NVIDIA RTX PRO 6000 Blackwell GPU.

\paragraph{Dataset}
As a prompt dataset, we use questions from \elifive{}~\citep{eli5}. 
For each watermark configuration, we use the same subset of prompts.
In \cref{app:additional_experiments:scenarios}, we instead use Chinese questions from the Zhihu-KOL dataset, keeping only prompts whose reference answer is longer than $200$ tokens.

\paragraph{Supported Payload Sizes}
ArcMark and Cycle-Shift are only evaluated with $8$- and $16$-bit bitstrings.
For both schemes, decoding enumerates all $2^m$ candidate messages, so the decoding time grows exponentially with the payload size $m$ and $32$-bit payloads cannot be decoded in a reasonable time.
Moreover, Cycle-Shift encodes the message as a cyclic shift of the vocabulary scores, which limits the payload to $\log_2 |\Sigma| \approx 17$ bits for \llama{}.

\subsection{Baseline Watermark Configuration}
\label{app:experimental_details:baselines_watermark_config}

Here, we describe the configuration used for our baseline watermarks. 
The notation used in this part is taken from each watermark's original paper.

\paragraph{ArcMark}
For ArcMark~\citep{arcmark}, we use the same default configuration as in~\citet{arcmark}. 
An $m$-bit message $M \in \{0,1\}^m$ is encoded as a codeword $C_M$ over $\mathbb{F}=\{0,\ldots,p-1\}$, and we set $p = 2^m$ so that the symbol alphabet matches the total number of possible messages. 
At each generation step $t$, the shared secret is $S_t=(V_t,\Pi_t)$, where $V_t \sim \mathrm{Unif}(\{0,\ldots,r-1\})$, and the channel input is formed as $z_t = \left(\frac{2\pi C_M(t)}{p} + \frac{2\pi V_t}{r} \right) \bmod 2\pi$. 
We set $r=256$ for the payload sizes we consider, implement $\Pi_t$ as an affine permutation over token IDs, and solve the optimal transport problem using the Sinkhorn algorithm with regularization coefficient $0.1$.
Because there is no official implementation for~\citet{arcmark}, we implement the ArcMark watermark algorithm from scratch.

\paragraph{BiMark}
For BiMark~\citep{bimark}, we use the same default configuration as in~\citet{bimark}. 
In particular, we use the base scaling factor $\tilde{\delta} = 1.0$ and the default number of layers $d = 10$. 
For the implementation, we fork the code from~\citet{bimark} and adapt it to work with vLLM.

\paragraph{Cycle-Shift}
For Cycle-Shift~\citep{three_bricks}, we use the AAR watermark~\citep{aar} as the underlying watermarking scheme, mirroring the default configuration from~\citet{three_bricks}.
Additionally, to obtain a distortionary version, we follow~\citet{gloaguen2026unifiedframeworkllmwatermarks} and vary $\delta$ in $\{0,0.1,0.2,0.3,0.4,0.5,0.75,1.0\}$.
For the implementation, we fork the code from~\citet{three_bricks} and adapt it to work with vLLM.

\paragraph{MC2Mark}
For MC2Mark~\citep{mc2mark}, we use the default number of layers $m = 10$.
For computational efficiency, we scale the number of segments $g$ with the number of bits.
In particular, we use $g=2$ segments for $16$ bits, $g=4$ for $32$ bits, and $g=8$ for $64$ bits.
Because there is no official implementation for~\citet{mc2mark}, we implement the MC2Mark watermark algorithm from scratch.

\paragraph{MPAC}
For MPAC~\citep{mpac}, we use Red-Green watermarks~\citep{kgw1} as the underlying watermarking scheme, mirroring the default configuration from~\citet{mpac}.
For the green list size, we use $\gamma = 0.5$, meaning that the green list size is $\gamma |\Sigma|$.
To enable faster inference, we slightly deviate from the original implementation of~\citet{kgw1} and use independent Bernoulli variables with probability $\gamma$ for the green list, where $1$ means that a given token is green. In contrast,~\citet{kgw1} enforce the green list size to be exactly $\gamma |\Sigma|$, which is equivalent to using correlated Bernoulli variables with probability $\gamma$ for the green list.
In terms of watermarking performance, the impact of this change is negligible~\citep{dlm_watermark}.
For the watermark strength, we vary $\delta$ in $\{1,2,3,4,5,6\}$.
For the implementation, we mostly reuse the code from~\citet{mpac}, but adjust it to work with vLLM.

\paragraph{MirrorMark}
For MirrorMark~\citep{mirrormark}, we use the Gumbel-max variant.
For CABS, we use the recommended configuration from~\citet{mirrormark}: we use frame size $f=3$, context window $W=4$, and max factor $\texttt{max\_factor} = 1.5$.
Because there is no official implementation for~\citet{mirrormark}, we implement the MirrorMark watermark algorithm from scratch.

\paragraph{RSBH}
For RSBH, we follow the recommended configuration from~\citet{rsbh} for the Reed--Solomon codes.
In particular, for a $16$-bit bitstring, we divide the original message into $4$ segments and then encode each segment using Reed--Solomon codes with $6$ segments and $4$ bits per segment. 
For a $32$-bit bitstring, we divide the original message into $4$ segments and then encode each segment using Reed--Solomon codes with $6$ segments and $8$ bits per segment.
For the frequency-balanced mapping, we use the realnewslike subset of the C4 dataset to estimate the frequencies.
For the watermark strength, we vary $\delta$ in $\{1,2,3,4,5,6\}$.
For the implementation, we mostly reuse the code from~\citet{rsbh}, but adjust it to work with vLLM.

\paragraph{StealthInk}
For StealthInk~\citep{stealthink}, we use one bit per position. 
Because there is no official implementation for~\citet{stealthink}, we implement the StealthInk watermark algorithm from scratch.

\subsection{Our Watermark Configuration}
\label{app:experimental_details:our_watermark_config}

For our watermark, to find $\lambda$ (\cref{eq:soft_ppl_wm}), we use a bisection algorithm on the interval $(0,100)$ with $60$ iterations.
To estimate $\mathbb{E}_{G}\left[ q(\tilde{G}'_t,p_t)\cdot \log p_t \right]$, we use a Monte Carlo estimation with $128$ samples.
For the watermark strength, we vary $\varepsilon$ in $\{0,0.1,0.2,0.3,0.5,0.7,0.9\}$.

\paragraph{Detection}
To detect the watermark, we follow the recommendation of~\citet{three_bricks}.
In particular, we de-duplicate the (context, token) pairs from a given text to ensure the statistical validity of our detection algorithm (see \cref{fig:zerobit_pvalue_calibration} in \cref{app:zerobit_detection}).

\section{Algorithms}
\label{app:algorithms}

In this section, we present the algorithmic description of our encoding (\cref{alg:encoding}) and decoding algorithms (\cref{alg:decoding}) from \cref{sec:method}.

\begin{algorithm}[H]
\caption{Encoding Multibit Algorithm}
\label{alg:encoding}
\begin{algorithmic}[1]
    \Require Prompt $\omega \in \Sigma^*$, Next-token probability $p_t \in \Delta(\Sigma)$, Message $M \in \{0,1\}^m$, Watermark private key $\xi$
    \State $s_t \gets \textsc{HashContext}(\xi,\omega_{<t})$ \Comment{Context is usually the 3 last tokens of $\omega_{<t}$}
    \State $G_t^1,\dots,G_t^m \in \{0,1\}^{|\Sigma|} \gets \textsc{PRNG}(s_t)$
    \For{$u \in \Sigma$}
        \For{$i \in \{1,\dots,m\}$}
            \State $\tilde{G}_t^i(u) \gets M_i G_t^i(u) + (1-M_i)(1-G_t^i(u))$ \Comment{Binomial encoding}
        \EndFor
        \State $\tilde{G}_t(u) \gets \sum_{i=1}^{m}\tilde{G}_t^i(u)$
    \EndFor
    \State $q_t \gets \textsc{WatermarkTransform}(p_t,\tilde{G}_t)$ \Comment{e.g., the PPL watermark (\cref{eq:soft_ppl_wm})}
    \State \Return $q_t$
\end{algorithmic}
\end{algorithm}

\begin{algorithm}[H]
\caption{Decoding Multibit Algorithm}
\label{alg:decoding}
\begin{algorithmic}[1]
    \Require Text $\omega \in \Sigma^*$, Message length $m \in \mathbb{N}$, Watermark private key $\xi$
    \State $S_1,\dots,S_m \gets 0$
    \For{$t \in \{1,\dots,|\omega|\}$}
        \State $s_t \gets \textsc{HashContext}(\xi,\omega_{<t})$ \Comment{Context is usually the 3 last tokens of $\omega_{<t}$}
        \State $G_t^1(\omega_t),\dots,G_t^m(\omega_t) \gets \textsc{PRNG}(s_t,\omega_t)$  \Comment{Same PRNG as in the encoding}
        \For{$i \in \{1,\dots,m\}$}
            \State $S_i \gets S_i + G_t^i(\omega_t)$
        \EndFor
    \EndFor
    \For{$i \in \{1,\dots,m\}$}
        \State $\hat{M}_i \gets \mathbf{1}\{S_i \ge |\omega|/2\}$ \Comment{Majority bit decoding. Ties are broken randomly}
    \EndFor
    \State $p_{\mathrm{value}} \gets \textsc{BinomialPValue}(S_1,\dots,S_m;|\omega|)$
    \State \Return $(\hat{M},p_{\mathrm{value}})$
\end{algorithmic}
\end{algorithm}

\section{Zero-Bit Watermark Detection}
\label{app:zerobit_detection}

In this section, we detail the zero-bit watermark detection we use in our work, and for prior works.
We find that some prior works have incorrectly calibrated p-values, leading to higher FPR than expected when running the watermark detection on human-generated text.

\begin{figure}
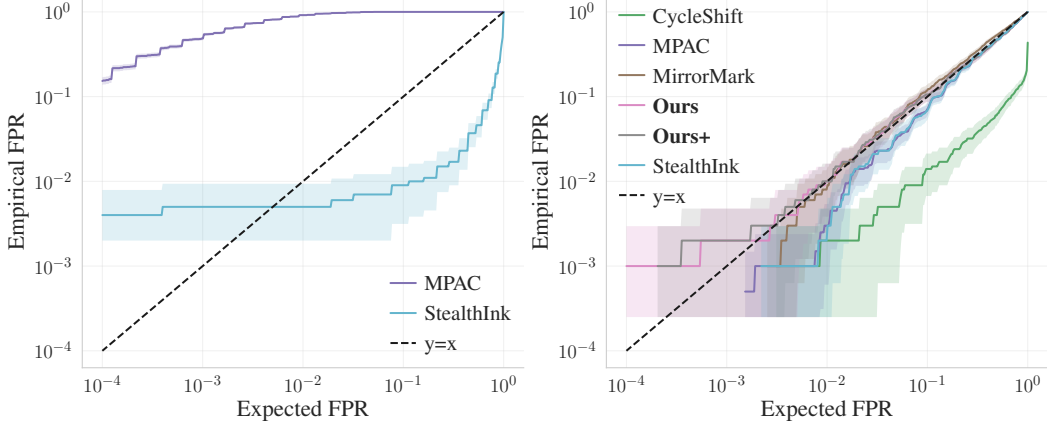

    \centering
    \begingroup
    \ifdefined\mathdefault\else\def\mathdefault#1{#1}\fi
    \resizebox{0.49\textwidth}{!}{\input{figures/appendix/zerobit/human_pvalue_original_null_check.pgf}}
    \resizebox{0.49\textwidth}{!}{\input{figures/appendix/zerobit/calibrated_pvalue_null_check.pgf}}
    \endgroup
    \caption{\textbf{Zero-Bit Watermark Detection Under the Null:} We show the empirical FPR versus the theoretical FPR for various zero-bit watermark detectors.
    The left side uses the original p-value implementation, whereas the right side uses Monte Carlo-based p-values (except for Cycle-Shift).
    To compute the p-values, we use $1000$ $200$-token-long texts extracted from the realnewslike subset of C4.
    For all watermarks, we use a bitstring length of $16$.
    }
    \label{fig:zerobit_pvalue_calibration}
\end{figure}

\paragraph{Our Zero-Bit Detector}
As mentioned in \cref{subsec:multibit_encoding_decoding}, our multibit watermark can also be used as a zero-bit watermark to decide whether a text $\omega$ is generated using our watermark or not.
We recall that $\hat{M}_i^t$ is the $i$-th decoded bit at token position $t$.
Let $S_i \coloneq \sum_{t=1}^{|\omega|} \hat{M}^t_i$ for each $i \in \{1,\dots,m\}$.
We propose computing the following statistic,
\begin{equation} \label{eq:loglikelihood_statistic}
    \Lambda(\omega) \coloneq \Lambda(S_1,\dots,S_m) \coloneq \sum_{i=1}^{m} \left[
        S_i \log\!\left(\frac{S_i}{|\omega|/2}\right)
        + (|\omega|-S_i)\log\!\left(\frac{|\omega|-S_i}{|\omega|/2}\right)
    \right].
\end{equation}
with the convention $0\log 0 \coloneq 0$. \cref{eq:loglikelihood_statistic} represents the difference in log-likelihood between the null hypothesis (all $S_i$ are Binomial $(|\omega|, 0.5)$) and the alternative that they are not Binomial.
In particular, $\Lambda$ is minimized when all $S_i = |\omega|/2$, which is the most likely case under the null.
Thus, to get a p-value, we use a Monte-Carlo estimation of
\begin{equation} \label{eq:our_pvalue}
    \text{p-value} \coloneq \mathbb{E}_{S_1,\dots,S_m}\!\left[\mathbf{1}\{\Lambda(S_1,\dots,S_m) \geq \Lambda(\omega)\}\right],
\end{equation}
where $S_1,\dots,S_m$ are sampled i.i.d.\ from Binomial$(|\omega|, 0.5)$.

\paragraph{Zero-Bit Detection Calibration}
To verify whether the zero-bit detection is well calibrated, we run the detection algorithm on $1000$ sequences of $200$ tokens extracted from human-generated text, in our case extracted from the C4 realnewslike dataset.
Then, we measure the number of texts flagged as watermarked (\ie empirical FPR) given an expected (theoretical) FPR.
If the p-values returned by the detector are sound, the empirical FPR should be below the theoretical FPR (\ie we should see a plot below the identity line).
In \cref{fig:zerobit_pvalue_calibration} (left), we show the empirical FPR given a theoretical FPR for two baselines, MPAC and StealthInk. 
We find that the suggested p-values are ill-calibrated.

\paragraph{Calibrate Zero-Bit Detectors}
For calibrated zero-bit p-values, we propose adjusting our Monte-Carlo approach from \cref{subsec:multibit_encoding_decoding} to prior works' statistics.
Specifically, for \citet{mpac,stealthink,mirrormark} we use the same statistic as the one suggested in their respective papers, but use Monte-Carlo estimation to estimate the null distribution.
For Cycle-Shift, we keep the original paper implementation as we find it is already calibrated and statistically sound.
In \cref{fig:zerobit_pvalue_calibration} (right), we show the empirical FPR given a theoretical FPR using the calibrated detector. 
We indeed find that, under the null, the empirical FPR is upper bounded by the theoretical FPR.
This means that the zero-bit p-values are now well-calibrated.

\section{Additional Experiments}
\label{app:additional_experiments}

\subsection{Zero-Bit Watermark Detection}
\label{app:additional_experiments:statistical_metrics}

In this part, we evaluate the zero-bit performance of our multibit watermark against relevant baselines.

\begin{figure}
    \centering
    \resizebox{\textwidth}{!}{\input{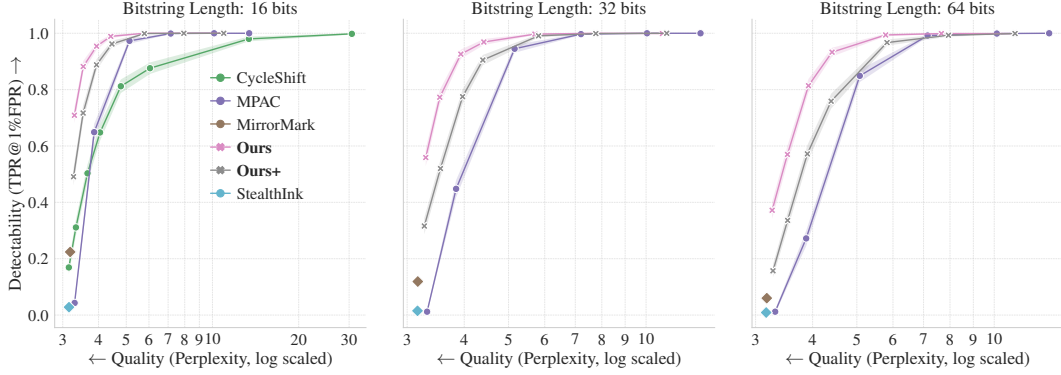}}
    \caption{\textbf{Comparison of the Detectability-Quality Trade-Off:} We compare the trade-off between zero-bit detectability (TPR@1\%FPR) and text quality ($\log \text{PPL}$) across various multibit watermarks and three bitstring lengths ($16$, $32$, and $64$ bits). 
    Each sample contains between $250$ and $350$ tokens and is generated by \llama{} using prompts from \elifive{}.}
    \label{fig:app_comparison_tpr}
\end{figure}

\paragraph{Zero-Bit Detection}
While zero-bit detection could be encoded in the embedded bitstring (\eg by using integrity bits), here we measure the zero-bit detection ability using the p-value provided by the watermark detector (\eg for our watermark, using \cref{eq:our_pvalue}).
For ArcMark, BiMark, MC2Mark, and RSBH, the watermark decoder does not provide p-values and we therefore exclude them from this analysis.

In particular, using the same experimental setup as in \cref{sec:evaluation}, in \cref{fig:app_comparison_tpr}, we measure on each generated sample the true positive rate at 1\% false positive rate (TPR@1\%FPR), \ie the average number of p-values below $0.01$.
We find that, for all watermarking schemes, the TPR@1\%FPR decreases with the payload size.
This means that, for a fixed impact on quality and number of tokens, it is harder to determine (using statistical tools) whether the decoded bitstring differs from random noise.
Regarding our watermark, the stateless version (Ours) (\ie using the scores $\tilde{G}$ from \cref{subsec:multibit_encoding_decoding}) outperforms all prior schemes, and importantly significantly outperforms the stateful version (Ours+) from \cref{subsec:multibit_improved}.
This result is expected: in the stateful version, we assign higher importance to bits that are currently not correctly decoded.
This means that, while all bits are more likely to be correctly decoded, none of them particularly stands out compared to random noise.

\subsection{Bit Accuracy}
\label{app:additional_experiments:bit_accuracy}

In this part, following prior works, we compare our scheme's bit accuracy with all baselines.
We find that, in most cases, our scheme outperforms prior works.

\begin{figure}
    \centering
    \resizebox{\textwidth}{!}{\input{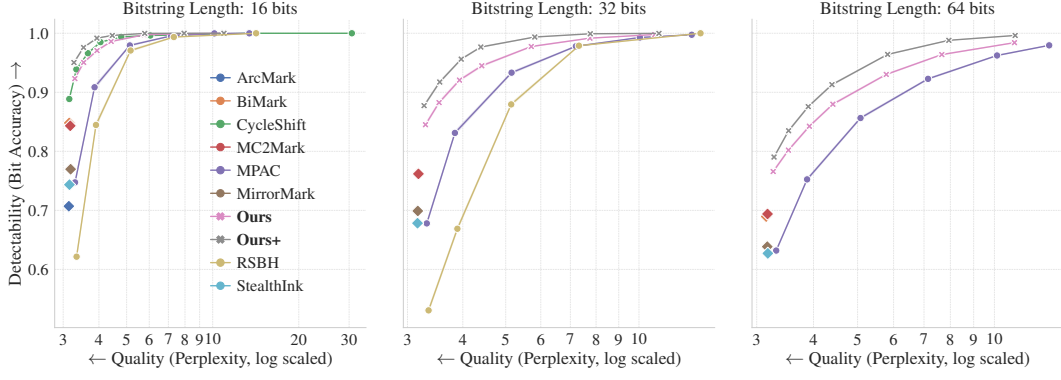}}
    \caption{\textbf{Comparison of the Bit Accuracy-Quality Trade-Off:} We compare the trade-off between bit accuracy and text quality ($\log \text{PPL}$) for various multibit watermarks and three bitstring lengths ($16$, $32$, and $64$ bits). 
    Each sample contains between $250$ and $350$ tokens and is generated by \llama{} using prompts from \elifive{}.
    For Cycle-Shift, only payloads below $\log_2 |\Sigma| \approx 17$ bits are supported.
    For ArcMark, only payloads below $16$ bits are supported.
    For RSBH, \citet{rsbh} does not provide a configuration for a $64$-bit payload.}
    \label{fig:main_comparison_bitacc}
\end{figure}    

\begin{wrapfigure}{r}{0.5\textwidth}
    \centering
    \resizebox{\linewidth}{!}{\input{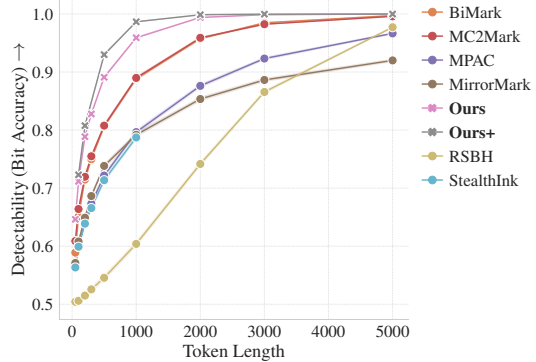}}
    \caption{\textbf{Bit Accuracy with Respect to the Number of Tokens:} We compare how the bit accuracy scales with the number of generated tokens.
        The hyperparameters of each watermarking scheme have been selected to have a similar impact on quality, and all schemes use a bit string of $32$ bits.}
    \label{fig:token_length_bitacc}
    \vspace{-0.3in}
\end{wrapfigure}

\paragraph{Improved Bit Accuracy}
For each completion, we record the bit accuracy, \ie the percentage of bits that we decode correctly.
A bit accuracy of $0.5$ means that the decoded message is random.
In \cref{fig:main_comparison_bitacc}, we show the bit accuracy of the different watermarking algorithms for different bitstring lengths.
For schemes with a strength parameter, we vary the strength to show the bit accuracy-quality trade-off, where we measure quality via perplexity.
Our approach achieves a better bit accuracy-quality trade-off than prior work, especially in low-distortion settings (\ie low-perplexity settings) and at larger bitstring lengths.

\paragraph{Scaling with the Number of Tokens}
In \cref{fig:token_length_bitacc}, we show how the bit accuracy scales with the number of tokens.
For this experiment, we select all watermark hyperparameters to operate in a low-distortion regime (\ie we use $\varepsilon = 0$ for our watermark), which ensures a similar quality impact across the compared methods, and we use a $32$-bit bitstring.
Across all tested baselines, our scheme scales more favorably with the number of tokens.
At $500$ tokens, the bit accuracy of our scheme exceeds $90$\%, whereas the best baseline achieves around $80$\%.

\subsection{Additional Quality Metrics}
\label{app:additional_experiments:quality_metrics}

In this part, we compare the message accuracy-quality trade-off of our watermark against the baselines using benchmark accuracy as a proxy for quality instead of perplexity.

\begin{figure}
    \centering
    \resizebox{\textwidth}{!}{\input{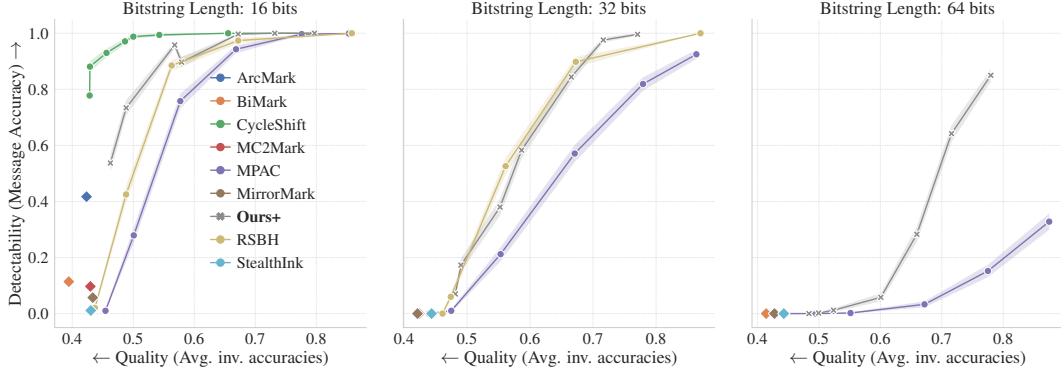}}
    \caption{\textbf{Comparison of the Message Accuracy-Quality Trade-Off using LLM Benchmarks:} We compare the trade-off between message accuracy and text quality (average benchmark accuracy) for various multibit watermarks and three bitstring lengths ($16$, $32$, and $64$ bits). 
    Each sample contains between $250$ and $350$ tokens and is generated by \llama{} using prompts from \elifive{}.
    For Cycle-Shift, only payloads below $\log_2 |\Sigma| \approx 17$ bits are supported.
    For ArcMark, only payloads below $16$ bits are supported.
    For RSBH, \citet{rsbh} does not provide a configuration for a $64$-bit payload.}
    \label{fig:lmeval_comparison_bitacc}
\end{figure}    

\begin{figure}
    \centering
    \providecommand{\mathdefault}[1]{#1}
    \resizebox{\textwidth}{!}{\input{figures/appendix/lm_eval/comparison_mer_main_lm_eval_by_benchmark_payload16.pgf}}
    \resizebox{\textwidth}{!}{\input{figures/appendix/lm_eval/comparison_mer_main_lm_eval_by_benchmark_payload32.pgf}}
    \resizebox{\textwidth}{!}{\input{figures/appendix/lm_eval/comparison_mer_main_lm_eval_by_benchmark_payload64.pgf}}
    \caption{\textbf{Message Accuracy-Quality Trade-Off on Individual Benchmarks:} We compare the trade-off between message accuracy and text quality, measured separately on GSM8K, HumanEval, MBPP, and MATH500 (as $1 - \text{accuracy}$), for various multibit watermarks and three bitstring lengths ($16$ bits, top; $32$ bits, middle; $64$ bits, bottom).
    Watermarked completions are generated by \llama{}.
    For Cycle-Shift, only payloads below $\log_2 |\Sigma| \approx 17$ bits are supported.
    For ArcMark, only payloads below $16$ bits are supported.
    For RSBH, \citet{rsbh} does not provide a configuration for a $64$-bit payload.}
    \label{fig:lmeval_by_benchmark}
\end{figure}

\paragraph{Experimental Setup}
For all watermark configurations, we measure benchmark accuracy using the EleutherAI evaluation harness~\citep{eval-harness} on GSM8k~\citep{gsm8k}, HumanEval~\citep{humaneval}, MBPP~\citep{mbpp}, and MATH500~\citep{math_dataset,lightman2024verify}.
In \cref{fig:lmeval_comparison_bitacc}, we report the average accuracy across benchmarks, and in \cref{fig:lmeval_by_benchmark}, we report the accuracy on each benchmark separately.

\paragraph{Improved Message Accuracy}
\cref{fig:lmeval_comparison_bitacc} shows that, when using benchmark accuracies as a proxy for quality, our approach still achieves a better message accuracy-quality trade-off than prior work.

\paragraph{Per-Benchmark Results}
\cref{fig:lmeval_by_benchmark} shows that the conclusions hold on each benchmark individually, including MATH500.
With a $32$-bit bitstring, Ours+ achieves the highest message accuracy at the strictest accuracy levels on all four benchmarks.
For instance, among configurations that preserve at least $40$\% HumanEval accuracy (respectively $60$\% GSM8K, $45$\% MBPP, and $25$\% MATH500 accuracy), Ours+ reaches $17.3$\% message accuracy, compared to $6.0$\% for the best baseline (RSBH).
At looser accuracy levels on HumanEval and GSM8K, RSBH, which relies on an error-correction code, can reach a higher message accuracy.
With a $16$-bit bitstring, Cycle-Shift achieves the best trade-off, consistent with \cref{fig:main_comparison_messacc}.

\paragraph{Unwatermarked Reference}
Without a watermark, \llama{} reaches $68.8$\% accuracy on GSM8K, $46.3$\% on HumanEval, $53.0$\% on MBPP, and $29.8$\% on MATH500, and its completions have a \qwen{} perplexity of $3.15$ in the setting of \cref{sec:evaluation}.
For the settings of \cref{app:additional_experiments:scenarios}, the unwatermarked perplexity is $6.55$ on Zhihu-KOL, $4.99$ with temperature $1$, and $3.17$ with top-k set to $25$.
Several of the evaluated schemes are distortion-free (ArcMark, BiMark, MC2Mark, MirrorMark, StealthInk, Cycle-Shift with $\delta = 0$, and Ours with SynthID from \cref{app:synthid_multibit}): in expectation over the watermark key, they sample from the same distribution as the unwatermarked model, and their perplexity and benchmark accuracies indeed match these reference values up to sampling noise.

\subsection{Additional Model}
\label{app:additional_experiments:mistral}

Here, we reproduce our main results using \mistral{} instead of \llama{}.

\begin{figure}
    \centering
    \resizebox{\textwidth}{!}{\input{figures/appendix/ministral/comparison_bit_accuracy_main_ministral.pgf}}
    \caption{\textbf{Comparison of the Bit Accuracy-Quality Trade-Off with \mistral{}:} We compare the trade-off between bit accuracy and text quality ($\log \text{PPL}$) for various multibit watermarks and three bitstring lengths ($16$, $32$, and $64$ bits). 
    Each sample contains between $250$ and $350$ tokens and is generated by \mistral{} using prompts from \elifive{}.
    For Cycle-Shift, only payloads below $\log_2 |\Sigma| \approx 17$ bits are supported.
    For ArcMark, only payloads below $16$ bits are supported.
    For RSBH, \citet{rsbh} does not provide a configuration for a $64$-bit payload.}
    \label{fig:ministral_comparison_bitacc}
\end{figure}    

\begin{figure}
    \centering
    \resizebox{\textwidth}{!}{\input{figures/appendix/ministral/comparison_mer_main_ministral.pgf}}
    \caption{\textbf{Comparison of the Message Accuracy-Quality Trade-Off with \mistral{}:} We compare the trade-off between message accuracy and text quality ($\log \text{PPL}$) for various multibit watermarks and three bitstring lengths ($16$, $32$, and $64$ bits). 
    Each sample contains between $250$ and $350$ tokens and is generated by \mistral{} using prompts from \elifive{}.}
    \label{fig:ministral_comparison_messacc}
\end{figure}    

\begin{figure}
    \centering
    \begingroup
    \ifdefined\mathdefault\else\def\mathdefault#1{#1}\fi
    \resizebox{\textwidth}{!}{\input{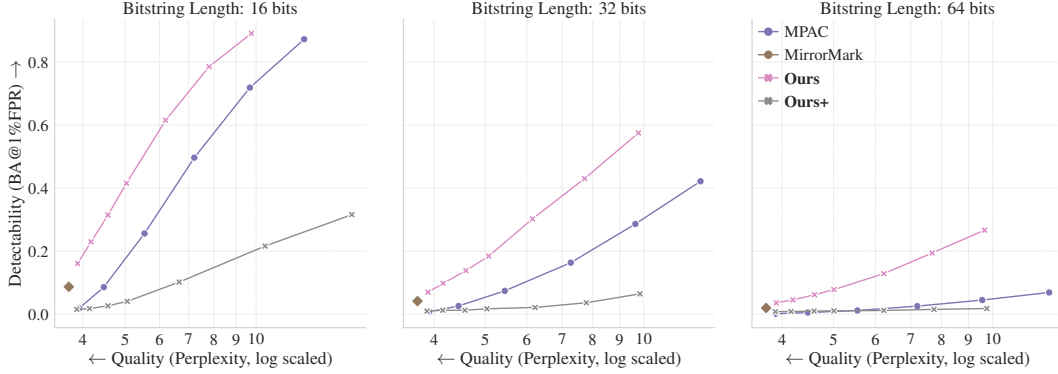}}
    \endgroup
    \caption{\textbf{Comparison of the Confidence Bit Accuracy-Quality Trade-Off with \mistral{}:} We compare the trade-off between confidence bit accuracy (BA@1\%FPR) and text quality ($\log \text{PPL}$) across several multibit watermarks and three bitstring lengths ($16$, $32$, and $64$ bits). 
    Each sample contains between $250$ and $350$ tokens and is generated by \mistral{} using prompts from \elifive{}.}
    \label{fig:bat_at_1_fpr_comparison_ministral}
\end{figure}

\paragraph{Experimental Setup}
We use the same experimental setup as in \cref{sec:evaluation}: for each watermark configuration, we generate completions of between $250$ and $350$ tokens for $1000$ prompts from \elifive{}, sample bitstrings uniformly at random, and measure quality using \qwen{} perplexity.
The only difference is that the watermarked text is generated with \mistral{} instead of \llama{}, using the inference parameters from \cref{app:experimental_details:inference_settings}.

\paragraph{Results}
In \cref{fig:ministral_comparison_bitacc,fig:ministral_comparison_messacc}, we observe the same trends as in \cref{fig:main_comparison_bitacc,fig:main_comparison_messacc}.
Our watermark achieves a better bit accuracy-quality trade-off than prior work, especially for larger bitstring lengths and in low-distortion regimes.
For message accuracy, our watermark also outperforms prior work at larger, more practical payloads ($32$ and $64$ bits).
As in the \llama{} evaluation, the ordering of the baselines differs between bit accuracy and message accuracy, since ArcMark, Cycle-Shift, and RSBH are designed to improve whole-message recovery rather than per-bit accuracy alone.
Finally, \cref{fig:bat_at_1_fpr_comparison_ministral} shows that the confidence-based BA@$1\%$ FPR metric follows the same pattern as with \llama{}: the stateless version of our watermark provides the best confidence bit accuracy-quality trade-off, whereas Ours+ underperforms on confidence-based metrics because it prioritizes the expected bit accuracy rather than per-bit confidence.

\subsection{Shorter Payloads}
\label{app:additional_experiments:payload8}

\begin{figure}
    \centering
    \providecommand{\mathdefault}[1]{#1}
    \resizebox{\textwidth}{!}{\input{figures/appendix/payload8/detectability_mainv2_8bit_multibit.pgf}}
    \caption{\textbf{Comparison of the Detectability-Quality Trade-Off with an $8$-bit Payload:} We compare the trade-off between BA@1\%FPR (left), bit accuracy (middle), and message accuracy (right) and text quality ($\log \text{PPL}$) for various multibit watermarks with an $8$-bit bitstring.
    Each sample contains between $250$ and $350$ tokens and is generated by \llama{} using prompts from \elifive{}.
    ArcMark, BiMark, and MC2Mark do not provide per-bit p-values, so their BA@1\%FPR is either omitted or shown as $0$.}
    \label{fig:payload8_comparison}
\end{figure}

In this part, we evaluate whether our watermark remains competitive for short bitstrings, where the benefit of encoding every bit at every position is expected to be smaller.

\paragraph{Experimental Setup}
We use the same experimental setup as in \cref{sec:evaluation}, but with an $8$-bit bitstring.
Because both ArcMark and Cycle-Shift support $8$-bit payloads, we compare against all baselines except RSBH.

\paragraph{Results}
\cref{fig:payload8_comparison} shows the same trends as for $16$-bit bitstrings.
Our stateless encoder (Ours) achieves the best BA@1\%FPR-quality trade-off: under a perplexity budget of $4$, it reaches $98.9$\% BA@1\%FPR, compared with $88.8$\% for MPAC and $84.7$\% for Cycle-Shift.
For message accuracy, Cycle-Shift performs best in the low-distortion regime ($98.3$\% under the same budget), closely followed by our stateful encoder (Ours+, $97.0$\%) and Ours ($96.2$\%).
Under a perplexity budget of $8$, all three schemes and MPAC exceed $99$\% message accuracy.
Hence, as expected, the advantage of our scheme shrinks at $8$ bits, but it retains the best confidence-aware trade-off.

\subsection{Other Inference Settings and Data Distributions}
\label{app:additional_experiments:scenarios}

\begin{figure}
    \centering
    \providecommand{\mathdefault}[1]{#1}
    \resizebox{\textwidth}{!}{\input{figures/appendix/scenarios/detectability_tradeoff_message_accuracy.pgf}}
    \caption{\textbf{Message Accuracy-Quality Trade-Off under Other Inference Settings and Data Distributions:} We compare the trade-off between message accuracy and text quality ($\log \text{PPL}$) for various multibit watermarks with a $32$-bit bitstring, when prompting \llama{} with Chinese questions from Zhihu-KOL (left), when sampling with temperature $1$ (middle), and when sampling with top-k set to $25$ (right).
    Each sample contains between $250$ and $350$ tokens.}
    \label{fig:scenarios_messacc}
\end{figure}

\begin{figure}
    \centering
    \providecommand{\mathdefault}[1]{#1}
    \resizebox{\textwidth}{!}{\input{figures/appendix/scenarios/detectability_tradeoff_ba_at_1pct_fpr.pgf}}
    \caption{\textbf{BA@1\%FPR-Quality Trade-Off under Other Inference Settings and Data Distributions:} We compare the trade-off between BA@1\%FPR and text quality ($\log \text{PPL}$) for various multibit watermarks with a $32$-bit bitstring, when prompting \llama{} with Chinese questions from Zhihu-KOL (left), when sampling with temperature $1$ (middle), and when sampling with top-k set to $25$ (right).
    Each sample contains between $250$ and $350$ tokens.
    BiMark, MC2Mark, and RSBH do not provide per-bit p-values, so their BA@1\%FPR is reported as $0$.}
    \label{fig:scenarios_bat_at_1}
\end{figure}

In this part, we verify that our results are not specific to the inference parameters and prompt distribution used in \cref{sec:evaluation}.

\paragraph{Experimental Setup}
We rerun the evaluation from \cref{sec:evaluation} with a $32$-bit bitstring in three independent scenarios, each changing a single aspect of the main setup: (i) a different data distribution, where we use Chinese questions from the Zhihu-KOL dataset instead of \elifive{} (keeping only prompts whose reference answer is longer than $200$ tokens), (ii) sampling with temperature $1$ instead of $0.7$, and (iii) sampling with top-k set to $25$ instead of $50$.
As in \cref{sec:evaluation}, we generate completions of between $250$ and $350$ tokens for $1000$ prompts with \llama{} and measure quality using \qwen{} perplexity.
We omit ArcMark and Cycle-Shift since they do not support $32$-bit payloads.

\paragraph{Results}
\cref{fig:scenarios_messacc,fig:scenarios_bat_at_1} show that, in all three scenarios, our conclusions from \cref{sec:evaluation} still hold: Ours+ achieves the best message accuracy-quality trade-off and Ours the best BA@1\%FPR-quality trade-off, in particular in the low-distortion regime.
For instance, on Zhihu-KOL with a perplexity budget of $8$, Ours reaches $89.1$\% BA@1\%FPR compared with $35.8$\% for the best baseline (MirrorMark), and Ours+ reaches $37.4$\% message accuracy while no baseline exceeds $0.2$\%.
With temperature $1$ and a perplexity budget of $6$, Ours reaches $93.6$\% BA@1\%FPR (MirrorMark: $71.6$\%) and Ours+ reaches $62.9$\% message accuracy (MC2Mark: $20.9$\%).
With top-k set to $25$ and a perplexity budget of $4$, Ours reaches $88.0$\% BA@1\%FPR (MPAC: $59.6$\%) and Ours+ reaches $40.4$\% message accuracy (RSBH: $4.8$\%).
At higher distortion, RSBH and MPAC catch up in message accuracy, consistent with \cref{fig:main_comparison_messacc}.

\subsection{Robustness-Quality Trade-Off}
\label{app:additional_experiments:robustness_quality}

\begin{table*}[t]
    \centering
    \caption{\textbf{Robustness-Quality Trade-Off:} Best bit accuracy (BA, top) and BA@1\%FPR (bottom), in \%, attainable by each watermark under a given perplexity budget after $20$\% word deletion, $20$\% synonym substitution, and paraphrasing with \textsc{GPT-5-mini}.
    For each watermark, we sweep the strength parameter and report the best value among configurations whose (unmodified) perplexity is below the budget.
    We use a $32$-bit bitstring; each sample contains between $250$ and $350$ tokens and is generated by \llama{} using prompts from \elifive{}.
    Missing values indicate that the scheme does not provide per-bit p-values.
    For word deletion and synonym substitution, the highest value in each column is \textbf{bolded}; under paraphrasing, all schemes are close to random guessing and we do not highlight any value.}
    \label{tab:robustness_quality_tradeoff}
    \renewcommand{\arraystretch}{1.15}
    \resizebox{\linewidth}{!}{%
    \begin{tabular}{lcccccccccccc}
    \toprule
     & \multicolumn{4}{c}{Word Deletion ($20$\%)} & \multicolumn{4}{c}{Synonym Substitution ($20$\%)} & \multicolumn{4}{c}{Paraphrasing} \\
    \cmidrule(lr){2-5}\cmidrule(lr){6-9}\cmidrule(lr){10-13}
      PPL budget & $\leq 4$ & $\leq 6$ & $\leq 8$ & $\leq 12$ & $\leq 4$ & $\leq 6$ & $\leq 8$ & $\leq 12$ & $\leq 4$ & $\leq 6$ & $\leq 8$ & $\leq 12$ \\
    \midrule
    \multicolumn{13}{l}{\emph{Bit accuracy (\%)}} \\
    BiMark & 61.2 & 61.2 & 61.2 & 61.2 & 61.5 & 61.5 & 61.5 & 61.5 & 50.9 & 50.9 & 50.9 & 50.9 \\
    MC2Mark & 61.7 & 61.7 & 61.7 & 61.7 & 61.7 & 61.7 & 61.7 & 61.7 & 50.1 & 50.1 & 50.1 & 50.1 \\
    MirrorMark & 52.1 & 52.1 & 52.1 & 52.1 & 52.0 & 52.0 & 52.0 & 52.0 & 50.0 & 50.0 & 50.0 & 50.0 \\
    StealthInk & 56.9 & 56.9 & 56.9 & 56.9 & 57.7 & 57.7 & 57.7 & 57.7 & 50.3 & 50.3 & 50.3 & 50.3 \\
    MPAC & 64.3 & 72.4 & 78.0 & 82.0 & 65.8 & 73.7 & 79.7 & 83.8 & 51.5 & 51.6 & 52.4 & 52.4 \\
    RSBH & 51.0 & 54.8 & 59.8 & 59.8 & 51.4 & 55.5 & 61.8 & 61.8 & 50.5 & 50.5 & 50.5 & 50.5 \\
    Ours & \textbf{71.3} & \textbf{79.4} & \textbf{83.7} & \textbf{88.1} & \textbf{72.4} & \textbf{80.6} & \textbf{85.3} & \textbf{89.8} & 51.5 & 52.0 & 52.0 & 52.0 \\
    Ours+ & 67.5 & 74.3 & 74.3 & 81.6 & 68.4 & 75.6 & 75.6 & 83.4 & 51.3 & 51.8 & 51.8 & 51.8 \\
    \midrule
    \multicolumn{13}{l}{\emph{BA@1\%FPR (\%)}} \\
    MirrorMark & 0.8 & 0.8 & 0.8 & 0.8 & 1.0 & 1.0 & 1.0 & 1.0 & 0.5 & 0.5 & 0.5 & 0.5 \\
    StealthInk & 0.5 & 0.5 & 0.5 & 0.5 & 0.7 & 0.7 & 0.7 & 0.7 & 0.6 & 0.6 & 0.6 & 0.6 \\
    MPAC & 2.9 & 13.5 & 35.1 & 51.6 & 3.1 & 15.9 & 42.7 & 62.8 & 0.6 & 0.7 & 0.7 & 0.7 \\
    Ours & \textbf{14.8} & \textbf{43.2} & \textbf{60.4} & \textbf{76.2} & \textbf{17.4} & \textbf{48.4} & \textbf{67.0} & \textbf{83.1} & 0.7 & 0.7 & 0.7 & 0.7 \\
    Ours+ & 2.5 & 11.9 & 11.9 & 35.7 & 3.0 & 12.1 & 12.1 & 41.2 & 0.4 & 0.7 & 0.7 & 0.8 \\
    \bottomrule
    \end{tabular}
    }
\end{table*}

In \cref{sec:evaluation:robustness}, we evaluate robustness in the low-distortion regime only.
Because different schemes have different impacts on quality, we now compare the robustness of all watermarks at matched quality.

\paragraph{Experimental Setup}
We use the same text modifications as in \cref{sec:evaluation:robustness} ($20$\% word deletion, $20$\% synonym substitution, and paraphrasing with \textsc{GPT-5-mini}) and a $32$-bit bitstring.
For each watermark with a strength parameter, we sweep the strength as in \cref{sec:evaluation} and apply every modification to the resulting completions.
For each scheme and perplexity budget, we then report the best robustness metric among all configurations whose perplexity (measured before modification) is below the budget.

\paragraph{Results}
\cref{tab:robustness_quality_tradeoff} shows that, after word deletion and synonym substitution, our stateless encoder (Ours) achieves the highest bit accuracy and BA@1\%FPR at every perplexity budget.
For instance, with a perplexity budget of $6$ and after $20$\% word deletion, Ours reaches $79.4$\% bit accuracy and $43.2$\% BA@1\%FPR, compared with $72.4$\% and $13.5$\% for the best baseline (MPAC).
Interestingly, while the stateful encoder (Ours+) is more robust than all baselines in the low-distortion regime (\cref{tab:robustness_eval_main}), the stateless encoder is more robust at matched quality.
Message accuracy remains low for all schemes after modification, with a maximum of $6.8$\% (Ours, after synonym substitution under a perplexity budget of $12$).
Under paraphrasing, none of the evaluated schemes is robust, even at higher watermark strengths: bit accuracy stays close to $50$\% and BA@1\%FPR stays below the $1$\% false-positive level.

\paragraph{Improving Robustness}
As with other watermarks based on local context hashing, most of the robustness loss comes from the modified context: after a word is deleted or substituted, the recomputed pseudorandom scores of the following tokens no longer match those used during generation.
This limitation is not specific to multibit watermarks, and our encoding is compatible with more robust seeding strategies, such as semantic hashing~\citep{semantic_watermark}, as they only change how the pseudorandom scores $G_t^1,\dots,G_t^m$ are seeded.

\subsection{Runtime}
\label{app:additional_experiments:runtime}

\begin{table}[t]
    \centering
    \caption{\textbf{Runtime Comparison:} Generation and detection throughput (tokens per second, higher is better) with a $16$-bit bitstring, as mean $\pm$ standard deviation over $5$ repeats.
    Slowdown is relative to generating without a watermark.
    Detection time is per $350$-token sample.
    \emph{(unc.)} denotes the unconstrained Soft PPL scheme from \cref{app:unconstrained_soft_ppl}.}
    \label{tab:throughput}
    \renewcommand{\arraystretch}{1.15}
    \resizebox{\linewidth}{!}{%
    \begin{tabular}{lrrrr}
    \toprule
      Watermark & Generation (tok/s) & Slowdown & Detection (tok/s) & Detection (ms/sample) \\
    \midrule
    No Watermark & $3{,}822 \pm 7$ & $1.0\times$ & -- & -- \\
    \midrule
    ArcMark & $196 \pm 5$ & $19.5\times$ & $4{,}558 \pm 48$ & $76.8 \pm 0.8$ \\
    BiMark & $503 \pm 17$ & $7.6\times$ & $3{,}115 \pm 138$ & $112.5 \pm 5.3$ \\
    Cycle-Shift & $1{,}834 \pm 9$ & $2.1\times$ & $45{,}070 \pm 899$ & $7.8 \pm 0.2$ \\
    MC2Mark & $265 \pm 5$ & $14.4\times$ & $1{,}223 \pm 74$ & $287.2 \pm 18.8$ \\
    MirrorMark & $200 \pm 3$ & $19.1\times$ & $676 \pm 13$ & $517.8 \pm 10.0$ \\
    MPAC & $1{,}833 \pm 13$ & $2.1\times$ & $718 \pm 3$ & $487.4 \pm 2.3$ \\
    RSBH & $2{,}389 \pm 223$ & $1.6\times$ & $5{,}023 \pm 271$ & $69.8 \pm 3.9$ \\
    StealthInk & $1{,}189 \pm 23$ & $3.2\times$ & $30{,}824 \pm 165$ & $11.4 \pm 0.1$ \\
    \midrule
    Ours & $472 \pm 8$ & $8.1\times$ & $5{,}134 \pm 81$ & $68.2 \pm 1.1$ \\
    Ours+ & $356 \pm 7$ & $10.7\times$ & $5{,}019 \pm 159$ & $69.8 \pm 2.3$ \\
    Ours (unc.) & $1{,}611 \pm 18$ & $2.4\times$ & $5{,}170 \pm 36$ & $67.7 \pm 0.5$ \\
    Ours+ (unc.) & $1{,}368 \pm 5$ & $2.8\times$ & $4{,}664 \pm 109$ & $75.1 \pm 1.7$ \\
    Ours (SynthID) & $158 \pm 3$ & $24.1\times$ & $4{,}581 \pm 211$ & $76.5 \pm 3.7$ \\
    Ours+ (SynthID) & $142 \pm 4$ & $26.8\times$ & $4{,}679 \pm 33$ & $74.8 \pm 0.5$ \\
    \bottomrule
    \end{tabular}
    }
\end{table}

\begin{table}[t]
    \centering
    \caption{\textbf{Overhead of the Stateful Encoder:} Generation throughput (tokens per second, mean $\pm$ standard deviation over $5$ repeats) of the stateless (Ours) and stateful (Ours+) encoders for different bitstring lengths, with the constrained Soft PPL scheme (\cref{eq:soft_ppl_wm}) and its unconstrained variant (\cref{app:unconstrained_soft_ppl}).
    The overhead column reports the additional time per generated token incurred by the stateful encoder.}
    \label{tab:stateful_overhead}
    \renewcommand{\arraystretch}{1.15}
    \resizebox{\linewidth}{!}{%
    \begin{tabular}{llrrrr}
    \toprule
      Soft PPL & Payload & Ours (tok/s) & Ours+ (tok/s) & Overhead (ms/token) & Slowdown \\
    \midrule
    Constrained & $16$ & $486 \pm 11$ & $364 \pm 5$ & $0.691$ & $1.34\times$ \\
     & $32$ & $580 \pm 2$ & $270 \pm 1$ & $1.982$ & $2.15\times$ \\
     & $64$ & $595 \pm 2$ & $179 \pm 0$ & $3.917$ & $3.33\times$ \\
    \midrule
    Unconstrained & $16$ & $1{,}609 \pm 12$ & $1{,}362 \pm 4$ & $0.113$ & $1.18\times$ \\
     & $32$ & $1{,}610 \pm 4$ & $1{,}362 \pm 3$ & $0.113$ & $1.18\times$ \\
     & $64$ & $1{,}607 \pm 5$ & $1{,}358 \pm 4$ & $0.114$ & $1.18\times$ \\
    \bottomrule
    \end{tabular}
    }
\end{table}

In this part, we measure the computational overhead of our watermark compared to unwatermarked generation and to the baselines.

\paragraph{Experimental Setup}
We load a single vLLM engine with \llama{} and time each watermarking scheme on the same workload: $64$ prompts from \elifive{}, generated with a fixed length of $350$ tokens and a $16$-bit bitstring, using the inference parameters from \cref{app:experimental_details:inference_settings}.
We repeat the measurement $5$ times on disjoint prompt subsets and report the mean and standard deviation.
For detection, we time the decoder on $16$ samples of $350$ tokens.

\paragraph{Generation and Detection Throughput}
\cref{tab:throughput} shows that the runtime of our watermark mostly depends on the underlying zero-bit scheme.
With the unconstrained Soft PPL scheme (\cref{app:unconstrained_soft_ppl}), which yields a similar detectability-quality trade-off to the constrained version, our watermark runs at $1{,}611$ tokens/s, which is comparable to the fastest baselines (RSBH, Cycle-Shift, and MPAC) and faster than the others.
The constrained version used in \cref{sec:evaluation} is slower ($472$ tokens/s), because it estimates an expectation over watermark scores and solves for $\lambda$ at every generation step.
With SynthID, the tournament with $n_{\text{layer}} = 100$ layers dominates the runtime.
Detection runs at around $5{,}000$ tokens/s for all our variants, \ie around $70$ ms per sample, which is faster than MPAC, MC2Mark, MirrorMark, and BiMark.
Detection with Cycle-Shift and StealthInk is faster, but for Cycle-Shift and ArcMark the decoding cost grows exponentially with the payload size, whereas for our scheme it grows linearly.
Lastly, since these measurements use an $8$B model, the relative overhead of all watermarks would be smaller for larger, more expensive models.

\paragraph{Overhead of the Stateful Encoder}
Compared to $\tilde{G}$, the stateful scores (\cref{eq:stateful_scores}) only require $|\mathcal{T}|$ evaluations of $\Phi$ per bit and candidate token.
\cref{tab:stateful_overhead} shows that, with the unconstrained Soft PPL scheme, the stateful encoder adds a constant overhead of around $0.11$ ms per generated token ($1.18\times$ slowdown) independently of the bitstring length.
With the constrained scheme, the overhead grows with the bitstring length (from $1.34\times$ at $16$ bits to $3.33\times$ at $64$ bits), because solving for $\lambda$ requires recomputing the stateful scores for each of the $128$ Monte Carlo samples (\cref{app:experimental_details:our_watermark_config}), whereas the stateless scores can be sampled directly.

\subsection{Ablation on the Stateful Encoder Remaining Bits}
\label{app:additional_experiments:ablation_remaining_bits}

In this part, we ablate the choice of the remaining-token parameter $T$ from the stateful encoder from \cref{subsec:multibit_improved}.

We recall that, in \cref{subsec:multibit_improved}, $T$ denotes the total number of generated tokens used in the stateful encoder's expected per-bit recovery estimate. Hence, at generation step $t$, the quantity $T - t$ is the number of future token positions remaining under the worst-case null assumption. Varying this value changes how strongly the stateful encoder prioritizes bits whose current decoding statistic $d_t^i$ is close to the decision boundary. In the main experiments, we use the averaged score over $\mathcal{T} \coloneq \{200, 300, 500, 1000, 2000\}$.

\paragraph{Experimental Setup}
To justify the averaged choice used in the main experiments, we ablate the values of $T$.
Specifically, we compute the message accuracy of our stateful encoder with $T \in \{64, 200, 300, 500, 1000, 2000\}$, as well as when averaging \cref{eq:stateful_scores} over $\mathcal{T}$ for different token lengths.

\begin{wrapfigure}{r}{0.5\textwidth}
    \vspace{-0.3in}
    \centering
    \resizebox{0.5\textwidth}{!}{\input{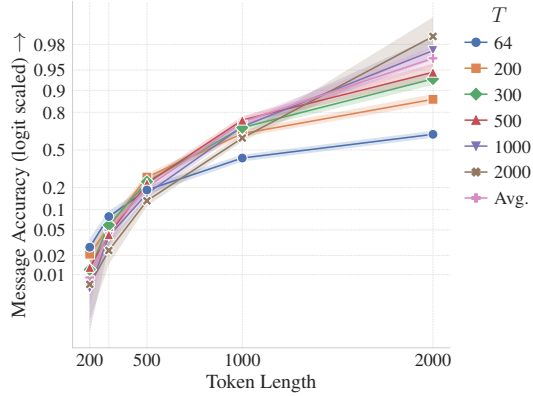}}
    \caption{\textbf{Ablation on the Remaining Bits $T$:} We compare the message accuracy of our stateful binomial encoder for different values of $T$ at different token lengths.}
    \label{fig:ablation_remaining_bits}
    \vspace{-0.3in}
\end{wrapfigure}

\paragraph{Results}
\cref{fig:ablation_remaining_bits} shows that, in most cases, using the same $T$ as the length of the generated text is optimal (\ie using $T = 2000$ when generating $2000$-token-long text).
When using the average, we find that it systematically remains close to the best-performing $T$.
The choice of $T$ has a significant effect on message accuracy: unlike bit accuracy, message accuracy is a whole-string metric, so moderate per-bit changes can translate into large differences in full-message recovery.
This result supports using the averaged score, since it avoids needing to estimate the generation length while retaining most of the gain from the best-matched $T$ values.

\subsection{Ablation on Position Allocation}
\label{app:additional_experiments:position_allocation}

\begin{wrapfigure}{r}{0.5\textwidth}
    \centering
    \resizebox{\linewidth}{!}{\input{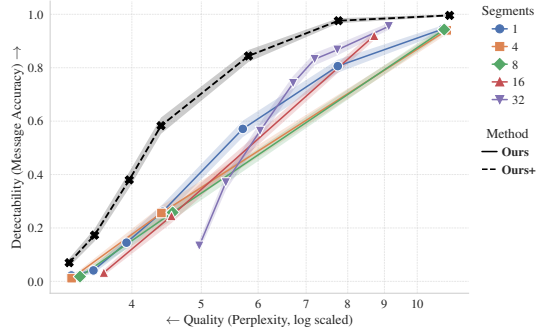}}
    \caption{\textbf{Message Accuracy for Different Segment Sizes:} We compare how our method combines with position allocation by splitting the $32$ bits into $k \in \{1,4,8,16,32\}$ segments.}
    \label{fig:mpac_comparison}
    \vspace{-0.4in}
\end{wrapfigure}

Here, we verify whether our method, which encodes every bit at once using binomial encoding, improves over position allocation.

\paragraph{Experimental Setup}
We use the same experimental setup as in \cref{sec:evaluation}: for each watermark configuration, we generate completions of between $250$ and $350$ tokens with \llama{} for $1000$ prompts from \elifive{}, sample bit strings uniformly at random, and measure quality using \qwen{} perplexity.
For the watermark, we combine binomial encoding with position allocation.
We split the $m$-bit string into $k$ segments.
Then, at each generation step, we use the context hash to select which segment to encode and use binomial encoding to encode the $m / k$ bits of this segment.
In particular, using $k \coloneq 1$ corresponds to using our watermark as described in \cref{sec:method}, and using $k \coloneq m$ corresponds to using MPAC with Soft-PPL as the underlying watermarking scheme.

\paragraph{Binomial Encoding Outperforms Position Allocation}
\cref{fig:mpac_comparison} shows the message accuracy--quality trade-off for various numbers of segments ($k \in \{1,4,8,16,32\}$), as well as for our method using the stateful encoder (Ours+).
We find that, with the stateless encoder, binomial encoding outperforms position allocation: using a smaller number of segments gives a better message accuracy--quality trade-off in relevant settings, but the advantage shrinks as the number of segments decreases.
Importantly, when using the stateful encoder (which is not compatible with position allocation), our method significantly outperforms position allocation.
This finding confirms our intuition that position allocation can be limiting for the power of multibit watermarking algorithms.
In \cref{app:theory:position_allocation}, we complement this ablation by proving that position allocation is never necessary: any position-allocation scheme can be converted into one where every bit influences every token, without losing quality or decoding power.

\section{SynthID-Text Multibit Watermark}
\label{app:synthid_multibit}

In this section, we propose replacing the Soft PPL watermarking scheme (\cref{eq:soft_ppl_wm}) with the SynthID-Text watermarking scheme from \citet{synthid}.

\paragraph{Zero-Bit SynthID-Text}
At each token position $t$, let $\omega_{<t} \in \Sigma^{t-1}$ be the current context and $p_t \in \Delta(\Sigma)$ the next-token probability distribution given by the LLM.
With SynthID-text, each token $u \in \Sigma$ is assigned a sequence of $n_{\text{layer}}$ Bernoulli scores $[G_t(u,1),\dots,G_t(u,n_{\text{layer}})]$ with probability $0.5$ called g-values.
Then, to get the watermarked probability distribution we recursively iterate the following equation,
\begin{equation} \label{eq:vectorized_tournament_layer}
   \forall u \in \Sigma, \quad p_t^{k}(u) = p_t^{k-1}(u) \left[1 + G_t(u,k) - \sum_{v\in \Sigma} p_t^k(v) G_t(v,k) \right],
\end{equation}
with $p_t^0 \coloneqq p_t$. 
The final watermarked probability distribution is given by $p_t^{n_{\text{layer}}}$, and each iteration of \cref{eq:vectorized_tournament_layer} is called a layer.
Importantly, this process is distortion-free, which means that in expectation over the scores, $p_t^{n_{\text{layer}}}$ is equal to $p_t$.

\paragraph{Multibit SynthID-Text}
We focus on a single layer first and drop the layer index from the notation.
For a single layer, the zero-bit version consists in assigning a Bernoulli score to each token and then applying \cref{eq:vectorized_tournament_layer}.
We simply propose replacing the Bernoulli scores with the complementary final score $\tilde{G}_t$ from \cref{subsec:multibit_encoding_decoding}.
To keep the single layer transformation distortion-free, we need to adjust \cref{eq:vectorized_tournament_layer}, as explained in \citet{gloaguen2026unifiedframeworkllmwatermarks}, 
\begin{equation} \label{eq:vectorized_tournament_layer_multibit}
   \forall u \in \Sigma, \quad q_t(\tilde{G},p_t)(u) = p_t(u) \left[1 + \sigma \left( \tilde{G}_t(u) - \sum_{v\in \Sigma} p_t(v) \tilde{G}_t(v) \right) \right],
\end{equation}
where $\sigma \coloneqq \left( \max_{v\in \Sigma} \tilde{G}_t(v) - \min_{v\in\Sigma} \tilde{G}_t(v) \right)^{-1}$.
For the $n_{\text{layer}}$ version, we simply iterate \cref{eq:vectorized_tournament_layer_multibit} using at each layer a different seed to sample the $m$ Bernoulli variables and encode the same bitstring $M$ with complementary scores (\cref{eq:bernoulli_transform}).

For detection, we adjust the weighted mean detector from \citet{synthid}, with
\begin{equation}
    \hat{M} \coloneqq \text{round}\!\left(
        \frac{1}{|\omega|\,n_{\text{layer}}}
        \sum_{t=1}^{|\omega|}
        \sum_{k=1}^{n_{\text{layer}}}
        w_k\,[G_t^1(\omega_t,k),\dots,G_t^m(\omega_t,k)]
    \right),
\end{equation}
where $w_1,\dots,w_{n_{\text{layer}}}$ are linearly decaying weights (from $10$ down to $1$) normalized such that $\sum_{k=1}^{n_{\text{layer}}} w_k = n_{\text{layer}}$.

\begin{figure}
    \centering
    \resizebox{\textwidth}{!}{\input{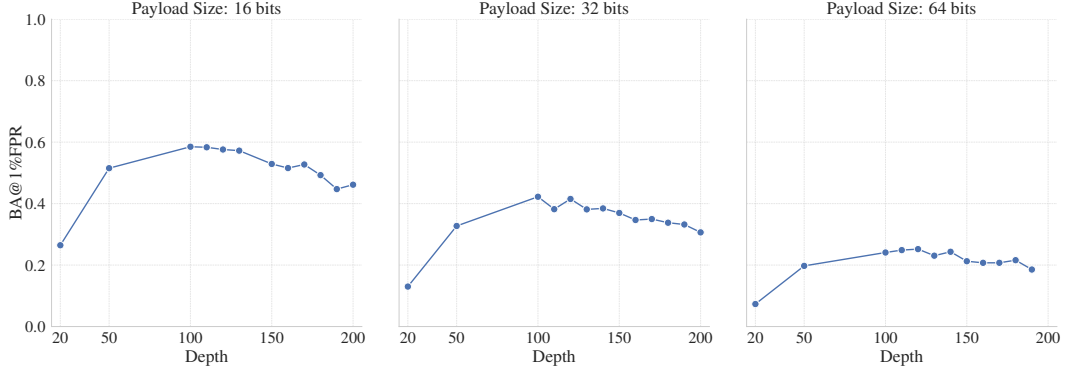}}
    \caption{\textbf{Comparison of the Detectability of Various SynthID Layers:} We compare the BA@1\%FPR for various numbers of SynthID layers $n_{\text{layer}}$ across three bitstring lengths ($16$, $32$, and $64$ bits).
    Each sample contains between $250$ and $350$ tokens and is generated by \llama{} using prompts from \elifive{}.}
    \label{fig:synthid_sweep}
\end{figure}

\begin{table*}[t]
    \centering
    \caption{\textbf{Distortion-Free Watermark Evaluation:} Message accuracy, BA@1\%FPR, and bit accuracy (BA) for 16-bit, 32-bit, and 64-bit bitstrings. All watermarks evaluated here are distortion-free, including Ours using SynthID. Missing values indicate that either the scheme does not support the metric or the bitstring length. Metrics are averaged over $1000$ samples where each sample contains between $250$ and $350$ tokens and is generated by \llama{} using prompts from \elifive{}.
    The highest value in each metric column is \textbf{bolded}.}
    \label{tab:distortion_free_metrics_main}
    \renewcommand{\arraystretch}{1.15}
    \resizebox{\linewidth}{!}{%
    \begin{tabular}{lccccccccc}
    \toprule
     & \multicolumn{3}{c}{16 bits} & \multicolumn{3}{c}{32 bits} & \multicolumn{3}{c}{64 bits} \\
    \cmidrule(lr){2-4}\cmidrule(lr){5-7}\cmidrule(lr){8-10}
      Watermark & Message Acc. & BA@1\%FPR & BA & Message Acc. & BA@1\%FPR & BA & Message Acc. & BA@1\%FPR & BA \\
    \midrule
    Cycle-Shift & \textbf{0.778} & -- & \textbf{0.889} & -- & -- & -- & -- & -- & -- \\
    ArcMark & 0.417 & -- & 0.707 & -- & -- & -- & -- & -- & -- \\
    BiMark & 0.114 & -- & 0.848 & 0.002 & -- & 0.762 & 0.000 & -- & 0.690 \\
    MC2Mark & 0.097 & -- & 0.843 & 0.000 & -- & 0.762 & 0.000 & -- & 0.694 \\
    MirrorMark & 0.057 & 0.077 & 0.770 & 0.000 & 0.034 & 0.699 & 0.000 & 0.019 & 0.639 \\
    StealthInk & 0.011 & -- & 0.744 & 0.000 & -- & 0.678 & 0.000 & -- & 0.627 \\
    \midrule
    Ours (SynthID) & 0.158 & \textbf{0.589} & 0.872 & \textbf{0.003} & \textbf{0.413} & \textbf{0.793} & 0.000 & \textbf{0.270} & \textbf{0.718} \\
    Ours+ (SynthID) & 0.045 & 0.446 & 0.816 & 0.000 & 0.322 & 0.749 & 0.000 & 0.190 & 0.687 \\
    \bottomrule
    \end{tabular}
    }
\end{table*}

\paragraph{Calibrating the Number of Layers}
In \cref{fig:synthid_sweep}, we compute the BA@1\%FPR for various numbers of tournament layers $n_{\text{layer}}$.
We find that using the multibit SynthID variant requires a higher number of layers $n_{\text{layer}}$ than the original zero-bit watermark, with detectability reaching a maximum around $n_{\text{layer}} = 100$.
We also find that using the stateful encoder from \cref{subsec:multibit_improved} (\ie replacing $\tilde{G}$ with $\tilde{G}'$) does not improve the bit accuracy.
We hypothesize that this occurs because the optimal number of layers depends on the distribution of the random scores.
With the stateful encoder, the distribution of scores is different at each step, but for detection we must keep a fixed number of tournament layers.
Therefore, the potential gains in accuracy from using the stateful encoder are negated by the loss incurred from using a suboptimal number of tournament layers.

\paragraph{Results: Our Scheme with SynthID Outperforms Most Distortion-Free Watermarks}
\cref{tab:distortion_free_metrics_main} shows the message accuracy, BA@1\%FPR, and bit accuracy for all considered distortion-free schemes, including ours with SynthID (using $n_{\text{layer}} = 100$).
We find that in most practical scenarios (\eg larger bitstrings), our scheme outperforms all prior baselines on all considered metrics.
At the same time, as anticipated, the stateful encoder does not improve message accuracy with SynthID compared to the stateless encoder.

\section{Unconstrained Soft PPL Scheme}
\label{app:unconstrained_soft_ppl}

\begin{figure}
    \centering
    \begingroup
    \ifdefined\mathdefault\else\def\mathdefault#1{#1}\fi
    \resizebox{\textwidth}{!}{\input{figures/appendix/unconstrained/comparison_mesacc_main_ppl_unconstrained.pgf}}
    \endgroup
    \caption{\textbf{Comparison of the Message Accuracy-Quality Trade-Off for Unconstrained Soft PPL:} We compare the trade-off between message accuracy and text quality ($\log \text{PPL}$) for the constrained Soft PPL parametrization (Ours) and the unconstrained parametrization using $\lambda$ as a scheme parameter (Ours (unc.)), for three bitstring lengths ($16$, $32$, and $64$ bits). Each sample contains between $250$ and $350$ tokens and is generated by \llama{} using prompts from \elifive{}.}
    \label{fig:unconstrained_comparison_messacc}
\end{figure}

\begin{figure}
    \centering
    \begingroup
    \ifdefined\mathdefault\else\def\mathdefault#1{#1}\fi
    \resizebox{\textwidth}{!}{\input{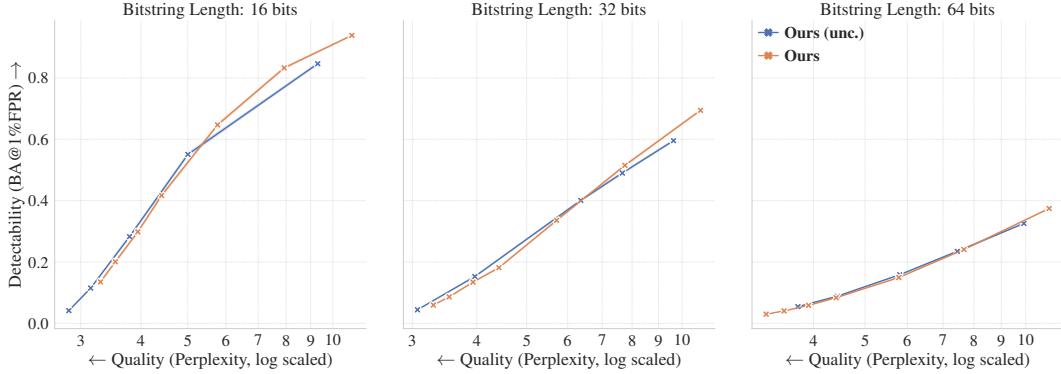}}
    \endgroup
    \caption{\textbf{Comparison of the Confidence Bit Accuracy-Quality Trade-Off for Unconstrained Soft PPL:} We compare the trade-off between confidence bit accuracy (BA@$1\%$ FPR) and text quality ($\log \text{PPL}$) for the constrained Soft PPL parametrization (Ours) and the unconstrained parametrization using $\lambda$ as a scheme parameter (Ours (unc.)), for three bitstring lengths ($16$, $32$, and $64$ bits). Each sample contains between $250$ and $350$ tokens and is generated by \llama{} using prompts from \elifive{}.}
    \label{fig:unconstrained_bat_at_1_fpr_comparison}
\end{figure}

In this section, we propose replacing the Lagrangian constraint from the Soft PPL watermarking scheme (\cref{eq:soft_ppl_wm}) with an unconstrained version.
We find that, while being computationally more efficient, this re-parametrization yields similar performance to the original one.

\paragraph{Experimental Setup} We use the same experimental setup as in \cref{sec:evaluation}: for each watermark configuration, we generate completions of between $250$ and $350$ tokens for $1000$ prompts from \elifive{}, sample bitstrings uniformly at random, and measure quality using \qwen{} perplexity. We compare the constrained version used in the main paper, where $\lambda$ is obtained by solving the Soft PPL constraint equation in \cref{eq:soft_ppl_wm} (effectively using $\varepsilon$ as a scheme parameter), against \emph{Ours (unc.)}, where $\lambda$ is varied directly as the scheme parameter.

\paragraph{Results} As shown in \cref{fig:unconstrained_comparison_messacc,fig:unconstrained_bat_at_1_fpr_comparison}, the constrained and unconstrained versions yield similar detectability-quality trade-offs across payload sizes. The unconstrained version is more computationally efficient because it avoids repeatedly estimating the expectation over watermark scores and solving for $\lambda$ at each generation step. However, it is harder to parametrize and compare across settings: a fixed $\varepsilon$ has the same expected quality impact across bitstring lengths by construction, whereas the quality impact of a fixed $\lambda$ depends on the bitstring length.
This is why we use the constrained version in our main experiments despite the increased computational cost.

\section{Broader Impacts and Licenses}

\subsection{Societal Impacts}
\label{app:societal_impacts}

\paragraph{Positive Impacts}
Multibit LLM watermarks can support provenance mechanisms for generated text.
Compared to zero-bit watermarks, they can encode richer metadata such as user identifiers, timestamps, or licensing information, which may help platform operators audit generated content, investigate misuse, and provide transparency around model-generated text.
Our work also emphasizes confidence-calibrated decoding, which can reduce overconfident interpretation of weak watermark signals.

\paragraph{Risks and Limitations}
The same ability to encode richer metadata can create privacy risks if identifiers are embedded without an appropriate governance policy, user notice, access control, or retention policy.
Watermark detectors may also be misused as evidence of authorship beyond their statistical scope, especially when text is edited, paraphrased, or generated by unwatermarked systems.
Our robustness experiments in \cref{sec:evaluation:robustness} show that multibit watermark recovery remains fragile under text modifications, so deployment should not rely on these methods as the only provenance or abuse-prevention mechanism.
Finally, a provider-controlled watermark can create asymmetric power: users may be unable to inspect what metadata is embedded or who can decode it.
Responsible deployment therefore requires clear disclosure, limited and purpose-specific payloads, key-management procedures, and calibration thresholds that reflect the intended decision risk.

\subsection{Existing Assets}
\label{app:societal_impacts:existing_assets}

\paragraph{Datasets}
\begin{itemize}
    \item \textbf{\elifive{}~\citep{eli5}.} We use questions from \elifive{} as prompts. The official repository for recreating \elifive{} is distributed under a BSD license for the accompanying software, but the dataset itself is derived from Reddit and CommonCrawl content.
    \item \textbf{C4 realnewslike.} We use the realnewslike subset of C4 to estimate frequencies for RSBH and to calibrate zero-bit detection on human-generated text. The Hugging Face dataset card lists C4 under ODC-BY and notes that use is also subject to the Common Crawl terms for the underlying content.
    \item \textbf{GSM8K~\citep{gsm8k}.} We use GSM8K through the EleutherAI evaluation harness as one of the benchmark-quality metrics. The Hugging Face dataset card lists the license as MIT.
    \item \textbf{HumanEval~\citep{humaneval}.} We use HumanEval through the EleutherAI evaluation harness as one of the benchmark-quality metrics. The official OpenAI repository is distributed under the MIT license.
    \item \textbf{MBPP~\citep{mbpp}.} We use MBPP through the EleutherAI evaluation harness as one of the benchmark-quality metrics. The Hugging Face dataset card for the Google Research dataset lists the license as CC-BY-4.0.
    \item \textbf{MATH500~\citep{math_dataset,lightman2024verify}.} We use MATH500 through the EleutherAI evaluation harness as one of the benchmark-quality metrics. It is a subset of the MATH dataset selected by~\citet{lightman2024verify}; both the MATH repository and the PRM800K repository that defines the subset are distributed under the MIT license.
    \item \textbf{Zhihu-KOL.} We use questions from the Zhihu-KOL dataset (\texttt{wangrui6/Zhihu-KOL} on Hugging Face) as prompts in \cref{app:additional_experiments:scenarios}. The dataset is derived from the Zhihu question-answering platform, and its Hugging Face dataset card does not list a license.
\end{itemize}

\paragraph{Models}
\begin{itemize}
    \item \textbf{\llama{}.} We use the instruction-tuned Llama 3.1 8B model for the main generation experiments. The Hugging Face model card lists the model under the Llama 3.1 Community License and requires gated access.
    \item \textbf{\mistral{}.} We use Ministral 3 14B Instruct for the additional-model experiments. The Hugging Face model card lists Apache-2.0 as the license.
    \item \textbf{\qwen{}.} We use Qwen3-30B-A3B to measure perplexity of generated samples. The Hugging Face model card lists Apache-2.0 as the license.
\end{itemize}

\section{Proof of Lemma \ref{lemma:expected_bit_acc_under_null}}
\label{app:proofs}

\begingroup
\renewcommand{\thelemma}{\ref{lemma:expected_bit_acc_under_null}}
\begin{lemma}
\ExpectedBitAccUnderNullStatement{}
\end{lemma}
\addtocounter{lemma}{-1}
\endgroup

\begin{proof}[Proof of Lemma \ref{lemma:expected_bit_acc_under_null}]
Fix a bit index $i \in \{1,\dots,m\}$, where $m$ is the message length in bits.
Recall that $\omega_{\le t} \in \Sigma^t$ is the generated prefix of length $t$, that $\omega_k$ is the realized token at position $k$, and that $\tilde{G}_k^i(\omega_k)\in\{0,1\}$ is the complementary score assigned to that realized token for bit $i$ at position $k$.
Also recall that
\begin{equation}
    d_t^i = 2 \sum_{k=1}^{t} \tilde{G}_k^i(\omega_k) - t
\end{equation}
is the signed decoding statistic of bit $i$ after the first $t$ generated tokens, and that $T$ is the total number of generated tokens, so $T-t$ is the number of future token positions.

For each future position $k \in \{t+1,\dots,T\}$, define the increment
\begin{equation}
    X_k^i \coloneq 2\tilde{G}_k^i(\omega_k)-1 \in \{-1,+1\}.
\end{equation}
Because the future tokens are assumed to be generated independently of the watermark, each future complementary score $\tilde{G}_k^i(\omega_k)$ is Bernoulli$(0.5)$.
Equivalently, the increments $X_{t+1}^i,\dots,X_T^i$ are i.i.d.\ Rademacher random variables with mean $0$ and variance $1$.
Therefore the final signed decoding statistic
\begin{equation}
    d_T^i = 2 \sum_{k=1}^{T} \tilde{G}_k^i(\omega_k) - T
\end{equation}
decomposes as
\begin{equation}
    d_T^i = d_t^i + \sum_{k=t+1}^{T} X_k^i.
\end{equation}

Conditioning on the current statistic $d_t^i$, we get
\begin{equation}
    P[d_T^i \ge 0 \mid d_t^i]
    =
    P\left[\sum_{k=t+1}^{T} X_k^i \ge -d_t^i \,\middle|\, d_t^i\right].
\end{equation}
Now $\sum_{k=t+1}^{T} X_k^i$ is the sum of exactly $T-t$ i.i.d.\ mean-zero, variance-one random variables, where $T-t$ is the number of future positions.
Hence, by the central limit theorem,
\begin{equation}
    \frac{1}{\sqrt{T-t}} \sum_{k=t+1}^{T} X_k^i \approx Z,
    \qquad Z \sim \mathcal{N}(0,1).
\end{equation}
Substituting this Gaussian approximation into the previous display yields
\begin{equation}
    P[d_T^i \ge 0 \mid d_t^i]
    \approx
    P\left[ Z \ge -\frac{d_t^i}{\sqrt{T-t}} \right]
    =
    \Phi\left(\frac{d_t^i}{\sqrt{T-t}}\right),
\end{equation}
where $\Phi$ is the CDF of the standard normal distribution.
This is exactly \cref{eq:expected_bit_accuracy_formula}.
\end{proof}

\section{Theoretical Analysis}
\label{app:theory}

In this section, we complement our empirical results with a theoretical analysis.
First, we compare encoding every bit at every position with position allocation (\cref{app:theory:position_allocation}).
Then, we show that the BA@$\alpha$\%FPR metric introduced in \cref{sec:considerations} cannot be computed from the zero-bit TPR@$\alpha$\%FPR and the bit accuracy (\cref{app:theory:ba_at_fpr}).

\subsection{Encoding Every Bit at Every Position versus Position Allocation}
\label{app:theory:position_allocation}

\paragraph{Position Allocation Is a Special Case of Weighted Binomial Encoding}
Consider the weighted generalization of our binomial scores from \cref{subsec:multibit_encoding_decoding},
\begin{equation} \label{eq:weighted_scores}
    \tilde{G}^{w}_t(u) \coloneq \sum_{i=1}^{m} w_{t,i}\, \tilde{G}_t^i(u), \qquad w_t \in \mathbb{R}_{\geq 0}^m,
\end{equation}
followed by any watermark transform $q(\tilde{G}^{w}_t, p_t)$.
Our stateless encoder uses $w_{t,i} = 1$ for all bits.
A position-allocation scheme that encodes only bit $J_t$ at position $t$ with strength $c_t$ corresponds to the one-sparse weights $w_t = c_t e_{J_t}$, where $e_j$ is the $j$-th standard basis vector.
The stateful encoder from \cref{subsec:multibit_improved} is also of this form: since $2\tilde{G}_t^i(u) - 1 \in \{-1,+1\}$, each term of \cref{eq:stateful_scores} is affine in $\tilde{G}_t^i(u)$, so that
\begin{equation}
    \tilde{G}'_t(u) = \sum_{i=1}^{m} w_{t,i}\, \tilde{G}_t^i(u) + C_t,
    \qquad
    w_{t,i} = \frac{1}{|\mathcal{T}|} \sum_{T \in \mathcal{T}} \left[ \Phi\!\left(\frac{d_{t-1}^i+1}{\sqrt{T}}\right) - \Phi\!\left(\frac{d_{t-1}^i-1}{\sqrt{T}}\right) \right] > 0,
\end{equation}
where $C_t$ does not depend on the candidate token $u$ and hence does not affect the watermark transforms we consider.
Hence, position allocation is the sparse boundary of a class of encoders that also contains both of our encoders, and, for a fixed watermark transform, optimizing over this class can never be worse than optimizing over position allocation.
This inclusion alone does not show that a scheme in which every bit influences every token can match any position-allocation scheme.
We show this next, without any assumption on the language model or on the watermark.

\paragraph{Setting}
We consider a fixed prompt and generations of $n$ tokens.%
\footnote{Variable-length generations can be handled by padding every text after the end-of-sequence token with a fixed dummy token.}
A multibit watermark is described by a secret key $\xi$, drawn independently of the message $M \in \{0,1\}^m$ from an arbitrary distribution, and by watermarked next-token distributions $q_t(\cdot \mid X_{<t}, \xi, M) \in \Delta(\Sigma)$ from which the tokens $X_1,\dots,X_n$ are sampled autoregressively.
The decoder observes the text $X = (X_1,\dots,X_n)$ and the key $\xi$, but not the sampling randomness.
For a prior $\pi$ over messages and a loss $\ell$ (\eg $\ell(M,\hat{M}) = \mathbf{1}\{M \neq \hat{M}\}$ for message accuracy, or the fraction of incorrect bits for bit accuracy), we measure the power of a scheme $W$ through its Bayes risk
\begin{equation}
    \mathcal{R}_{\pi,\ell}(W) \coloneq \inf_{d} \mathbb{E}_{M \sim \pi}\left[ \ell\left(M, d(\xi, X)\right) \right],
\end{equation}
where the infimum is over all decoders $d$.
A scheme uses \emph{position allocation} if there exist indices $J_t = J_t(X_{<t}, \xi) \in \{1,\dots,m\}$ such that each $q_t(\cdot \mid X_{<t}, \xi, M)$ depends on $M$ only through $M_{J_t}$.
We say that bit $i$ \emph{influences} position $t$ if the conditional distribution of $X_t$ given $(X_{<t}, \text{key}, M)$ changes with positive probability when $M_i$ is flipped.

\begin{theorem}[Encoding every bit at every position loses neither power nor quality] \label{theorem:dense_dominance}
    Let $W$ be any multibit watermark, in particular any position-allocation scheme.
    Then there exists a multibit watermark $V$, with key $(\xi, \xi')$, such that
    \begin{enumerate}
        \item for every message $M$ and key $\xi$, the text generated by $V$ has the same conditional distribution given $(\xi, M)$ as the text generated by $W$. In particular, $V$ and $W$ generate identically distributed texts and have the same quality under any quality metric,
        \item $V$ is at least as powerful as $W$: $\mathcal{R}_{\pi,\ell}(V) \leq \mathcal{R}_{\pi,\ell}(W)$ for every prior $\pi$ and every loss $\ell$,
        \item every bit influences every position $t$ at which $q_t(\cdot \mid X_{<t}, \xi, M)$ has at least two tokens in its support.
    \end{enumerate}
\end{theorem}

\begin{proof}
Let $\xi' = (U_{t,m'})_{t \leq n,\, m' \in \{0,1\}^m}$ be i.i.d.\ $\mathrm{Uniform}(0,1)$ variables, independent of $(\xi, M)$.
Fix an ordering $v_1,\dots,v_{|\Sigma|}$ of the vocabulary, and for $q \in \Delta(\Sigma)$ and $u \in (0,1)$ let $F_q^{-1}(u) \coloneq v_k$ with $k$ the smallest index such that $\sum_{j \leq k} q(v_j) \geq u$, so that $F_q^{-1}(U) \sim q$ for $U \sim \mathrm{Uniform}(0,1)$.
We define $V$ as the scheme that generates
\begin{equation}
    X_t = F^{-1}_{q_t(\cdot \mid X_{<t}, \xi, M)}\left(U_{t,M}\right), \qquad t = 1,\dots,n,
\end{equation}
\ie $V$ replaces the sampling randomness of $W$ at position $t$ with the key entry indexed by the position and the full message.

(1) Fix $\xi$ and $M$.
The prefix $X_{<t}$ only depends on $U_{1,M},\dots,U_{t-1,M}$, so $U_{t,M}$ is uniform and independent of $X_{<t}$, and $X_t \sim q_t(\cdot \mid X_{<t}, \xi, M)$ conditionally on $X_{<t}$.
By the chain rule, $X$ has the same conditional distribution given $(\xi, M)$ under $V$ as under $W$.
Since the key $\xi$ has the same distribution in both schemes, $(\xi, X)$ has the same distribution given $M$ under both schemes, and so does the text $X$ alone.

(2) Let $d$ be any decoder for $W$, and let $d'(\xi, \xi', X) \coloneq d(\xi, X)$ be the decoder for $V$ that ignores $\xi'$.
By (1), $\mathbb{E}[\ell(M, d'(\xi,\xi',X))]$ under $V$ equals $\mathbb{E}[\ell(M, d(\xi,X))]$ under $W$.
Taking the infimum over $d$ gives $\mathcal{R}_{\pi,\ell}(V) \leq \mathcal{R}_{\pi,\ell}(W)$.

(3) Fix $t$, $X_{<t}$, $\xi$, $M$, and a bit $i$, and let $M'$ be $M$ with bit $i$ flipped.
Given the full key $(\xi, \xi')$, $X_t$ is the deterministic token $F^{-1}_{q}(U_{t,M})$ under message $M$ and $F^{-1}_{q'}(U_{t,M'})$ under $M'$, where $q$ and $q'$ are the corresponding watermarked distributions.
Since $U_{t,M}$ and $U_{t,M'}$ are independent, these tokens coincide with probability $\sum_{v \in \Sigma} q(v) q'(v) \leq \max_{v} q(v) < 1$ whenever $q$ has at least two tokens in its support.
Hence, flipping $M_i$ changes the conditional distribution of $X_t$ with positive probability.
\end{proof}

\paragraph{Interpretation}
\cref{theorem:dense_dominance} shows that position allocation is never necessary: every position-allocation scheme, including the best one, can be converted into a scheme in which every bit influences every token without changing the generated text distribution and without losing decoding power, for any metric (\eg bit accuracy, message accuracy, or BA@$\alpha$\%FPR).
In practice, the uniform variables $U_{t,M}$ are generated pseudorandomly from the key, the position, and the message, as for any watermark score.
The condition in (3) holds for schemes that keep sampling randomness, such as Red-Green watermarks, as soon as the next-token distribution $p_t$ has at least two tokens in its support, since $q_t \propto p_t \exp(\delta G_t)$ has the same support as $p_t$.

The theorem does not show that encoding every bit at every position is \emph{strictly} more powerful, nor that our binomial encoding specifically achieves such a gain.
The constructed scheme $V$ also generally carries additional information about the message, but exploiting it requires enumerating candidate messages, which is intractable for long payloads.
Whether binomial encoding improves over position allocation is therefore an empirical question, which the segment ablation in \cref{app:additional_experiments:position_allocation} answers positively: at matched perplexity, encoding more bits jointly (\ie using fewer segments) improves the message accuracy-quality trade-off.

\subsection{BA@\texorpdfstring{$\alpha$}{alpha}\%FPR Is Not Determined by TPR@\texorpdfstring{$\alpha$}{alpha}\%FPR and Bit Accuracy}
\label{app:theory:ba_at_fpr}

A natural question is whether the BA@$\alpha$\%FPR metric from \cref{sec:considerations} can be derived from two metrics commonly reported in prior work: the zero-bit TPR@$\alpha$\%FPR (\cref{app:additional_experiments:statistical_metrics}) and the bit accuracy on watermarked text.
We show that this is not the case.

For a given decoded bit, let $C$ be the event that it is decoded correctly by majority vote, and $D_\alpha$ the event that its two-sided binomial p-value is at most $\alpha$.
Sampling a watermarked text and a bit, we have
\begin{equation}
    \mathrm{BA} = P(C), \qquad \text{BA@}\alpha\text{\%FPR} = P(C \cap D_\alpha), \qquad \text{TPR@}\alpha\text{\%FPR} = P(Z_\alpha),
\end{equation}
where $Z_\alpha$ is the rejection event of the zero-bit detector (\cref{eq:our_pvalue}).
Hence, BA is a marginal probability, whereas BA@$\alpha$\%FPR is a joint probability.

\begin{proposition} \label{proposition:ba_at_fpr}
    There is no function $f$ such that $\text{BA@}\alpha\text{\%FPR} = f\left(\text{TPR@}\alpha\text{\%FPR}, \mathrm{BA}\right)$ for all distributions of watermarked text.
    This holds even for a single-bit payload ($m = 1$), where the zero-bit rejection event $Z_\alpha$ coincides with the per-bit rejection event $D_\alpha$.
\end{proposition}

\begin{proof}
Even if $Z_\alpha = D_\alpha$, with $x \coloneq P(D_\alpha)$ and $y \coloneq P(C)$, the Fréchet bounds only give $\max\{0, x + y - 1\} \leq P(C \cap D_\alpha) \leq \min\{x, y\}$, and every value in this interval is attained by some joint distribution of $(C, D_\alpha)$.
We give an explicit example.
Let $m = 1$, $M_1 = 1$, $|\omega| = 8$ tokens, and let $S$ be the number of votes for $1$, so that the bit is correct if $S > 4$.
Under the null $S \sim \mathrm{Binomial}(8, 1/2)$, the two-sided p-values are $p(S=0) = p(S=8) = 2/2^8 \approx 0.0078$ and $p(S=1) = p(S=7) = 18/2^8 \approx 0.0703$.
At $\alpha = 1\%$, among these outcomes only $S \in \{0, 8\}$ is significant.
Since, for $m = 1$, the statistic $\Lambda$ in \cref{eq:loglikelihood_statistic} is a symmetric and increasing function of $|S - 4|$, the zero-bit detector rejects exactly when $S \in \{0, 8\}$ as well.
Consider two distributions of the watermarked statistic:
\begin{center}
\begin{tabular}{lcccc}
    \toprule
    & $P(S=8)$ & $P(S=7)$ & $P(S=1)$ & $P(S=0)$ \\
    \midrule
    Distribution A & $0.4$ & $0.4$ & $0.1$ & $0.1$ \\
    Distribution B & $0.5$ & $0.3$ & $0.2$ & $0$ \\
    \bottomrule
\end{tabular}
\end{center}
Both have $\mathrm{BA} = P(S \in \{7,8\}) = 0.8$ and $\text{TPR@}1\text{\%FPR} = P(S \in \{0,8\}) = 0.5$, but $\text{BA@}1\text{\%FPR} = P(S = 8)$ is $0.4$ under A and $0.5$ under B.
\end{proof}

The difference between the two distributions is the probability of a confident but incorrect decoding ($S = 0$), which neither BA nor TPR captures.
Such distributions arise, for instance, as mixtures over prompts with different entropy.
For $m > 1$, the zero-bit detector aggregates evidence across all bits, so TPR@$\alpha$\%FPR is even less related to the per-bit quantity $P(C \cap D_\alpha)$.
BA@$\alpha$\%FPR can be derived from TPR and BA only under additional assumptions, \eg that every vote is an i.i.d.\ Bernoulli variable with the same known success probability across all samples, which does not hold for LLM-generated text.
Hence, BA@$\alpha$\%FPR must be measured directly.

\fi

\newpage
\section*{NeurIPS Paper Checklist}

\begin{enumerate}

\item {\bf Claims}
    \item[] Question: Do the main claims made in the abstract and introduction accurately reflect the paper's contributions and scope?
    \item[] Answer: \answerYes{}
    \item[] Justification: We present our method in \cref{sec:method} that introduces binomial encoding. In every token, we encode the full message through this binomial variable. Our evaluation in \cref{sec:evaluation} matches the description and claims from the abstract.

\item {\bf Limitations}
    \item[] Question: Does the paper discuss the limitations of the work performed by the authors?
    \item[] Answer: \answerYes{}
    \item[] Justification: The paper discusses limitations in \cref{sec:conclusion}, including the limited robustness of multibit watermarks under text modifications and the absence of theoretical guarantees that non-position allocation methods are strictly better.

\item {\bf Theory assumptions and proofs}
    \item[] Question: For each theoretical result, does the paper provide the full set of assumptions and a complete (and correct) proof?
    \item[] Answer: \answerYes{}
    \item[] Justification: The paper states the assumptions for \cref{lemma:expected_bit_acc_under_null} in \cref{subsec:multibit_improved}, and provides its proof in \cref{app:proofs}.

    \item {\bf Experimental result reproducibility}
    \item[] Question: Does the paper fully disclose all the information needed to reproduce the main experimental results of the paper to the extent that it affects the main claims and/or conclusions of the paper (regardless of whether the code and data are provided or not)?
    \item[] Answer: \answerYes{}
    \item[] Justification: The main setup is described in \cref{sec:evaluation}, and \cref{app:experimental_details} gives the model, dataset, inference settings, watermark configurations, and detection details needed to reproduce the reported comparisons. Our code is publicly available at \url{https://github.com/eth-sri/multibit-binomial-wm}.

\item {\bf Open access to data and code}
    \item[] Question: Does the paper provide open access to the data and code, with sufficient instructions to faithfully reproduce the main experimental results, as described in supplemental material?
    \item[] Answer: \answerYes{}
    \item[] Justification: The code needed to reproduce the paper is publicly available at \url{https://github.com/eth-sri/multibit-binomial-wm}. All the datasets/models used are open-source.

\item {\bf Experimental setting/details}
    \item[] Question: Does the paper specify all the training and test details (e.g., data splits, hyperparameters, how they were chosen, type of optimizer) necessary to understand the results?
    \item[] Answer: \answerYes{}
    \item[] Justification: The evaluation setup is specified in \cref{sec:evaluation}, including prompt count, generation length, decoding metrics, and robustness edits. Further model, dataset, inference, baseline, and proposed-watermark hyperparameters are given in \cref{app:experimental_details}.

\item {\bf Experiment statistical significance}
    \item[] Question: Does the paper report error bars suitably and correctly defined or other appropriate information about the statistical significance of the experiments?
    \item[] Answer: \answerYes{}
    \item[] Justification: Whenever relevant, our figures include 95\% CI.

\item {\bf Experiments compute resources}
    \item[] Question: For each experiment, does the paper provide sufficient information on the computer resources (type of compute workers, memory, time of execution) needed to reproduce the experiments?
    \item[] Answer: \answerYes{}
    \item[] Justification: We state in \cref{app:experimental_details:inference_settings} we use vLLM on a single NVIDIA RTX PRO 6000 Blackwell GPU for inference.
    
\item {\bf Code of ethics}
    \item[] Question: Does the research conducted in the paper conform, in every respect, with the NeurIPS Code of Ethics \url{https://neurips.cc/public/EthicsGuidelines}?
    \item[] Answer: \answerYes{}
    \item[] Justification: The work is a study of multibit LLM watermarks using public datasets and existing models, and does not involve human-subject experiments, private data collection, or deployment decisions.

\item {\bf Broader impacts}
    \item[] Question: Does the paper discuss both potential positive societal impacts and negative societal impacts of the work performed?
    \item[] Answer: \answerYes{}
    \item[] Justification: The paper discusses positive provenance and auditability impacts, as well as privacy, misuse, over-reliance, and robustness risks in \cref{app:societal_impacts}.
    
\item {\bf Safeguards}
    \item[] Question: Does the paper describe safeguards that have been put in place for responsible release of data or models that have a high risk for misuse (e.g., pre-trained language models, image generators, or scraped datasets)?
    \item[] Answer: \answerNA{}
    \item[] Justification: The paper does not release a new pretrained model, image generator, or scraped dataset with high misuse risk; it proposes and evaluates a watermarking algorithm.

\item {\bf Licenses for existing assets}
    \item[] Question: Are the creators or original owners of assets (e.g., code, data, models), used in the paper, properly credited and are the license and terms of use explicitly mentioned and properly respected?
    \item[] Answer: \answerYes{}
    \item[] Justification: The paper cites the models, datasets, software, and baseline methods used in \cref{sec:evaluation,app:experimental_details}; \cref{app:societal_impacts:existing_assets} lists the licenses and terms for the datasets and models used in the experiments.

\item {\bf New assets}
    \item[] Question: Are new assets introduced in the paper well documented and is the documentation provided alongside the assets?
    \item[] Answer: \answerNA{}
    \item[] Justification: The paper introduces a new algorithm, but it does not currently release a new dataset, model checkpoint, or other standalone asset requiring asset documentation.

\item {\bf Crowdsourcing and research with human subjects}
    \item[] Question: For crowdsourcing experiments and research with human subjects, does the paper include the full text of instructions given to participants and screenshots, if applicable, as well as details about compensation (if any)? 
    \item[] Answer: \answerNA{}
    \item[] Justification: The paper does not involve crowdsourcing experiments or research with human subjects.

\item {\bf Institutional review board (IRB) approvals or equivalent for research with human subjects}
    \item[] Question: Does the paper describe potential risks incurred by study participants, whether such risks were disclosed to the subjects, and whether Institutional Review Board (IRB) approvals (or an equivalent approval/review based on the requirements of your country or institution) were obtained?
    \item[] Answer: \answerNA{}
    \item[] Justification: The paper does not involve crowdsourcing experiments or research with human subjects.

\item {\bf Declaration of LLM usage}
    \item[] Question: Does the paper describe the usage of LLMs if it is an important, original, or non-standard component of the core methods in this research? Note that if the LLM is used only for writing, editing, or formatting purposes and does \emph{not} impact the core methodology, scientific rigor, or originality of the research, declaration is not required.
    \item[] Answer: \answerYes{}
    \item[] Justification: The paper studies watermarking for LLM generation and describes the LLMs used for generation and evaluation in \cref{sec:evaluation,app:experimental_details}.

\end{enumerate}

\end{document}